\documentclass[11pt]{article}
\usepackage[parfill]{parskip}
\usepackage[utf8x]{inputenc}
\usepackage{ucs}
\usepackage[english]{babel}
\usepackage{amsthm,amsmath,amssymb,latexsym,stmaryrd}
\usepackage{amsfonts}
\usepackage[all]{xy}
\usepackage{graphicx}
\usepackage{scrextend,enumitem, hyperref}
\makeatletter
\def\namedlabel#1#2{
   \begingroup #2%
   \def\@currentlabel{#2}%
   \phantomsection\label{#1}\endgroup
}
\makeatother

\author{Linan Chen\and Florence Clerc\and Prakash Panangaden}
\title{Bisimulation for Feller-Dynkin Processes}

\makeatletter
\newcommand{\singlespacing}{\let\CS=\@currsize\renewcommand{\baselinestretch}{1}\small\CS}
\newcommand{\doublespacing}{\let\CS=\@currsize\renewcommand{\baselinestretch}{1.75}\small\CS}
\newcommand{\normalspacing}{\let\CS=\@currsize\renewcommand{\baselinestretch}{\BLS}\small\CS}
\makeatother

\theoremstyle{definition}
\newtheorem{thm}{Theorem}[section]

\newtheorem{cor}[thm]{Corollary}

\newtheorem{prop}[thm]{Proposition}

\newtheorem{remark}[thm]{Remark}
 
\newtheorem{theorem}[thm]{Theorem}

\newtheorem{lemma}[thm]{Lemma}
\newtheorem{proposition}[thm]{Proposition}

\newtheorem{defin}[thm]{Definition}

\newtheorem{definition}[thm]{Definition}

 \def\pushright#1{{%            % set up
    \parfillskip=0pt            % so \par doesnt push \square to left
    \widowpenalty=10000         % so we dont break the page before \square
    \displaywidowpenalty=10000  % ditto
    \finalhyphendemerits=0      % TeXbook exercise 14.32
    \leavevmode                 % \nobreak means lines not pages
    \unskip                     % remove previous space or glue
    \nobreak                    % don't break lines
    \hfil                       % ragged right if we spill over
    \penalty50                  % discouragement to do so
    \hskip.2em                  % ensure some space
    \null                       % anchor following \hfill
    \hfill                      % push \square to right
    {#1}                        % the end-of-proof mark (or whatever)
    \par}}                      % build paragraph
 \def\qed{\pushright{\rule{2mm}{3mm}}\penalty-700 \smallskip}

\renewenvironment{proof}{\begin{trivlist} \item[{\bf ~Proof}.]}%
 {\qed\end{trivlist}}

\makeatletter

\newdimen\w@dth

\def\setw@dth#1#2{\setbox\z@\hbox{\scriptsize $#1$}\w@dth=\wd\z@
\setbox\@ne\hbox{\scriptsize $#2$}\ifnum\w@dth<\wd\@ne \w@dth=\wd\@ne \fi
\advance\w@dth by 1.2em}

\def\t@^#1_#2{\allowbreak\def\n@one{#1}\def\n@two{#2}\mathrel
{\setw@dth{#1}{#2}
\mathop{\hbox to \w@dth{\rightarrowfill}}\limits
\ifx\n@one\empty\else ^{\box\z@}\fi
\ifx\n@two\empty\else _{\box\@ne}\fi}}
\def\t@@^#1{\@ifnextchar_ {\t@^{#1}}{\t@^{#1}_{}}}

\def\t@left^#1_#2{\def\n@one{#1}\def\n@two{#2}\mathrel{\setw@dth{#1}{#2}
\mathop{\hbox to \w@dth{\leftarrowfill}}\limits
\ifx\n@one\empty\else ^{\box\z@}\fi
\ifx\n@two\empty\else _{\box\@ne}\fi}}
\def\t@@left^#1{\@ifnextchar_ {\t@left^{#1}}{\t@left^{#1}_{}}}

\def\two@^#1_#2{\def\n@one{#1}\def\n@two{#2}\mathrel{\setw@dth{#1}{#2}
\mathop{\vcenter{\hbox to \w@dth{\rightarrowfill}\kern-1.7ex
                 \hbox to \w@dth{\rightarrowfill}}%
       }\limits
\ifx\n@one\empty\else ^{\box\z@}\fi
\ifx\n@two\empty\else _{\box\@ne}\fi}}
\def\tw@@^#1{\@ifnextchar_ {\two@^{#1}}{\two@^{#1}_{}}}

\def\tofr@^#1_#2{\def\n@one{#1}\def\n@two{#2}\mathrel{\setw@dth{#1}{#2}
\mathop{\vcenter{\hbox to \w@dth{\rightarrowfill}\kern-1.7ex
                 \hbox to \w@dth{\leftarrowfill}}%
       }\limits
\ifx\n@one\empty\else ^{\box\z@}\fi
\ifx\n@two\empty\else _{\box\@ne}\fi}}
\def\t@fr@^#1{\@ifnextchar_ {\tofr@^{#1}}{\tofr@^{#1}_{}}}

\newdimen\W@dth
\def\setW@dth#1#2{\setbox\z@\hbox{$#1$}\W@dth=\wd\z@
\setbox\@ne\hbox{$#2$}\ifnum\W@dth<\wd\@ne \W@dth=\wd\@ne \fi
\advance\W@dth by 1.2em}

\def\T@^#1_#2{\allowbreak\def\N@one{#1}\def\N@two{#2}\mathrel
{\setW@dth{#1}{#2}
\mathop{\hbox to \W@dth{\rightarrowfill}}\limits
\ifx\N@one\empty\else ^{\box\z@}\fi
\ifx\N@two\empty\else _{\box\@ne}\fi}}
\def\T@@^#1{\@ifnextchar_ {\T@^{#1}}{\T@^{#1}_{}}}

\def\T@left^#1_#2{\def\N@one{#1}\def\N@two{#2}\mathrel{\setW@dth{#1}{#2}
\mathop{\hbox to \W@dth{\leftarrowfill}}\limits
\ifx\N@one\empty\else ^{\box\z@}\fi
\ifx\N@two\empty\else _{\box\@ne}\fi}}
\def\T@@left^#1{\@ifnextchar_ {\T@left^{#1}}{\T@left^{#1}_{}}}

\def\Tofr@^#1_#2{\def\N@one{#1}\def\N@two{#2}\mathrel{\setW@dth{#1}{#2}
\mathop{\vcenter{\hbox to \W@dth{\rightarrowfill}\kern-1.7ex
                 \hbox to \W@dth{\leftarrowfill}}%
       }\limits
\ifx\N@one\empty\else ^{\box\z@}\fi
\ifx\N@two\empty\else _{\box\@ne}\fi}}
\def\T@fr@^#1{\@ifnextchar_ {\Tofr@^{#1}}{\Tofr@^{#1}_{}}}

\def\Two@^#1_#2{\def\N@one{#1}\def\N@two{#2}\mathrel{\setW@dth{#1}{#2}
\mathop{\vcenter{\hbox to \W@dth{\rightarrowfill}\kern-1.7ex
                 \hbox to \W@dth{\rightarrowfill}}%
       }\limits
\ifx\N@one\empty\else ^{\box\z@}\fi
\ifx\N@two\empty\else _{\box\@ne}\fi}}
\def\Tw@@^#1{\@ifnextchar_ {\Two@^{#1}}{\Two@^{#1}_{}}}

\def\to{\@ifnextchar^ {\t@@}{\t@@^{}}}
\def\from{\@ifnextchar^ {\t@@left}{\t@@left^{}}}
\def\two{\@ifnextchar^ {\tw@@}{\tw@@^{}}}
\def\tofro{\@ifnextchar^ {\t@fr@}{\t@fr@^{}}}
\def\To{\@ifnextchar^ {\T@@}{\T@@^{}}}
\def\From{\@ifnextchar^ {\T@@left}{\T@@left^{}}}
\def\Two{\@ifnextchar^ {\Tw@@}{\Tw@@^{}}}
\def\Tofro{\@ifnextchar^ {\T@fr@}{\T@fr@^{}}}

\makeatother

\newcommand{\cF}{\mathcal{F}}
\newcommand{\cG}{\mathcal{G}}

\newcommand{\cS}{\mathcal{S}}

\begin{document}
\maketitle

\begin{abstract}
  Bisimulation is a concept that captures behavioural equivalence.  It has
  been studied extensively on nonprobabilistic systems and on discrete-time
  Markov processes and on so-called continuous-time Markov chains.  In the
  latter time is continuous but the evolution still proceeds in jumps.  We
  propose two definitions of bisimulation on continuous-time stochastic
  processes where the evolution is a \emph{flow} through time.  We show that they
  are equivalent and we show that when restricted to discrete-time, our
  concept of bisimulation encompasses the standard discrete-time concept.
  The concept we introduce is not a straightforward generalization of
  discrete-time concepts.
\end{abstract}

\section{Introduction} 
Bisimulation~\cite{Milner80,Park81,Sangiorgi09} is a fundamental concept in
the theory of transition systems capturing a strong notion of behavioural
equivalence.  In particular, it is a notion stronger than that of trace
equivalence.  Bisimulation has been widely studied for discrete time
systems where transitions happen as steps, both on discrete~\cite{Larsen91}
and continuous state spaces~\cite{Blute97,Desharnais02,Panangaden09}.  In
all these types of systems a crucial ingredient of the definition of
bisimulation is the ability to talk about \emph{the next step}.  Thus, the
general format of the definition of bisimulation is that one has some
property that must hold ``now'' (in the states being compared) and then one
says that the relation is preserved in the next step.

Some attempts have been made to talk about
continuous-time~\cite{Desharnais03b}, but even in what are called
continuous-time Markov chains there is a discrete notion of time
\emph{step}; it is only that there is a real-valued duration associated
with each state that makes such systems continuous time.  They are often
called ``jump processes'' in the mathematical literature, see, for example,
\cite{Rogers00a,Whitt02}, a phrase that better captures the true nature of
such processes.

%The extension to
%\emph{continuous-time} Markov chains~\cite{Desharnais03b} has also been
%considered.  

Outside of computer science, there is a vast range of systems that involve
true continuous-time evolution: deterministic systems governed by
differential equations and stochastic systems governed by ``noisy''
differential equations called stochastic differential equations.  These
have been extensively studied for over a century since the pioneering work
of Einstein~\cite{Einstein1905} on Brownian motion.  In the computer
science literature there have been studies of very special systems that
feature continuous time: timed automata~\cite{Alur94} and hybrid
systems~\cite{Alur95}.  In these systems the time evolution is assumed to
be piecewise constant (timed automata) or piecewise smooth (hybrid
automata) and bisimulation is defined without recourse to talking about the
next step.  However, a general formalism that covers processes like
diffusion is not available as far as we are aware.

In this work we aim at a general theory of bisimulation for stochastic
systems with true continuous-time evolution.  We focus on a class of
systems called Feller-Dynkin processes for which a good mathematical theory
exists.  These systems are the most general version of Markov processes
defined on continuous state spaces and with continuous time evolution.
Such systems encompass Brownian motion and its many variants.

The obvious extension of previous definitions of bisimulation on discrete
Markov processes or on jump processes fail to provide a meaningful notion
of behavioural equivalence as we will illustrate later on.  It is a mistake
to think that one can get a reasonably good understanding of such systems
by considering suitable ``limits'' of discrete-time systems.  Intuitively,
the notion of bisimulation is sensitive to small changes that are not
captured when taking the limit.  It is true that, for example, Brownian
motion can be seen as arising as a limit, in the sense of convergence in
distribution, of a discrete random walk as both the discrete time unit and
the step size go to zero.  However, entirely new phenomena occur with the
trajectories of the Brownian motion which are not understandable through
the limiting process at least not in any naive sense: the probability of
being at any single state $x$ at a given time $t$ is zero, but the
probability of hitting $x$ before a given time $s$ is strictly positive.

%Not surprisingly, it becomes much harder to compute
%anything in examples and the literature is replete with techniques to try
%and get a handle on the behaviour of such systems.

%Consider something as simple as the concept of the \emph{second} time that
%a Brownian motion hits a set.  This cannot be defined in any simple way.
%The first time that a set is hit can be readily defined but the jittering
%of a trajectory of the Brownian motion near the boundary of the set means that we cannot clearly tell
%when it exits the set, so it becomes impossible to clearly define hitting a
%set the second time.  Indeed there is a notion called \emph{local time}
%which measures the accumulative time that such a process spends \emph{at a
%  point}; clearly a concept that is meaningless in the discrete-time
%continuous-space case.

To avoid those issues, we work with the set of trajectories of the system.
A number of possible ways had to be explored and in the end the particular
version we present here turned out to have the desired properties: (a)
corresponds to our intuition in a number of examples and (b) correctly
specializes to the discrete-time case.

Section \ref{sec:FDprocesses} explains the mathematical background on
Feller-Dynkin processes and Brownian motion. In section
\ref{sec:bisimulation}, we show why a naive extension of previous
definition of bisimulation does not work and we propose a new definition of
bisimulation as an equivalence relation that we illustrate on a number of
examples. In section \ref{sec:FDcospan}, we give an equivalent definition
of bisimulation as a cospan of morphisms extending the previous notion of
span of ``zig-zag'' morphisms in the discrete-time case. In section
\ref{sec:DT}, we show that our definition of bisimulation is coherent with
the previous definition of bisimulation in discrete time. Much remains to
be done, of course, as we describe in the concluding section.

\section{Background on Feller-Dynkin processes}
\label{sec:FDprocesses}
We assume that basic concepts like topology, measure theory and basic
concepts of probability on continuous spaces are well known; see, for
example~\cite{Billingsley08,Dudley89,Panangaden09}.  

The basic arena for the action is a probability space.
\begin{definition}
A \emph{probability space} is a triple $(S,\cF,P)$ where $S$ is a space
(usually some kind of topological space), $\cF$ is a $\sigma$-algebra
(usually its Borel algebra) and $P$ is a probability measure on $\cF$.
\end{definition}

Given a measurable space $(X,\Sigma)$ a Markov kernel is a map
$\tau:X\times\Sigma\to [0,1]$ which is measurable in its first argument,
i.e. $\tau(\cdot,A\in\Sigma):X\to \mathbb{R}$ is measurable for any fixed
$A$ in $\Sigma$ and for any fixed $x\in X$, $\tau(x,\cdot)$ is a
(sub)probability measure.  These kernels describe transition probability
functions. 

A crucial concept is that of a \emph{filtration}.  They will play a central
role in the description of a process.
\begin{defin}
  A \emph{filtration} on a measurable space $(\Omega,\cF)$ is a
  nondecreasing family $(\cF_t)_{t \geq 0}$ of sub-$\sigma$-algebras of $\cF$,
  \emph{i.e.}\ $\cF_s \subseteq \cF_t \subseteq \cF$
  for $0\leq s < t < \infty$.
\end{defin}
This concept is used to capture the idea that at time $t$ what is ``known'' or
``observed'' about the process is encoded in the sub-$\sigma$-algebra
$\cF_t$.

\begin{defin}
A \emph{stochastic process} is a collection of random variables $(X_t)_{0
  \leq t < \infty}$ on a measurable space $(\Omega, \cF)$ that take values
in a second measurable space $(S, \cS)$ called the \emph{state space}.  We
say that a stochastic process is \emph{adapted} to a filtration $\cF_t$ if
for each $t\geq 0$ we have $X_t$ is $\cF_t$-measurable.
\end{defin}
Note that a stochastic process is always adapted to the filtration $\cG_t$,
where $\cG_t$ is defined as the $\sigma$-algebra generated by all the
random variables $\{X_s| s\leq t\}$. The filtration $(\cG_t)_{t\geq 0}$ is also referred to as the natural filtration associated to $(X_t)_{t \geq 0}$.

\smallskip
Before stating the definition of the continuous-time processes we will be interested in, let us first start by recalling the definition of their discrete-time counterparts.
\begin{defin}
A labelled Markov process (LMP) is a triple $(X,\Sigma,\tau)$ where
$(X,\Sigma)$ is a measurable space and $\tau$ is a Markov kernel.  
\end{defin}
%These look different from stochastic processes but it is well-known that
%Markov processes are special cases of stochastic processes, see any
%standard textbook, such as~\cite{Billingsley08}.

We will quickly review the theory of continuous-time processes on
continuous state space; much of this material is adapted from ``Diffusions,
Markov Processes and Martingales, Volume I'' by Rogers and
Williams~\cite{Rogers00a} and we use their notations.  Another useful source
is ``Functional analysis for probability and stochastic processes'' by
A. Bobrowski~\cite{Bobrowski05}.  Let $E$ be a locally compact Hausdorff
space with countable base and let it be equipped with the Borel
$\sigma$-algebra $\mathcal{E} = \mathcal{B}(E)$.  $E_\partial$ is the
one-point compactification of $E$: $E_\partial = E \uplus \{ \partial \}$.
The physical picture is that the added state, $\partial$, represents a
point at infinity; we will view it as an absorbing state.

We say that a continuous real-valued function $f$ on $E$ ``vanishes at
infinity'' if for every $\varepsilon > 0$ there is a compact subset
$K \subset E$ such that $\forall x\in E\setminus K$ we have
$|f(x)| \leq \varepsilon$.  This space is a Banach space with the $\sup$
norm.  

\begin{definition}
A \emph{semigroup} of operators on any Banach space is a family
of linear continuous (bounded) operators $T_t$ indexed by
$t\in\mathbf{R}^{\geq 0}$ such that
\[ \forall s,t \geq 0, T_s \circ T_t = T_{s+t}\]
and
\[ T_0 = I. \]
\end{definition}
The first equation above is called the semigroup property.  The operators
in a semigroup are continuous however there is a useful continuity property
of the semigroup as a whole.
\begin{definition}
For $X$ a Banach space, we say that a semigroup $T_t:X\to X$ is \emph{strongly continuous} if 
\[ \forall x\in X, \lim_{t\downarrow 0} T_t x = x\]
which is equivalent to saying
\[ \forall x\in X, \lim_{t\downarrow 0}\| T_t x - x \| \to 0. \]
\end{definition}

\begin{definition}
A \emph{Feller-Dynkin semigroup (FDS)} is a strongly continuous
semigroup $(\hat{P}_t)_{t \geq 0}$ of linear operators on $C_0(E)$ (the
space of continuous functions on $E$ which vanish at infinity) satisfying the additional condition: 
\[\forall t \geq 0 ~~~ \forall f \in C_0(E) \text{, if }~~ 0 \leq f \leq 1 \text{, then }~~ 0 \leq \hat{P}_t f \leq 1\]
\end{definition}

The following important proposition relates these FDS with Markov
processes which allows one to see the connection with more familiar
probabilistic transition systems.
\begin{prop}
Given such an FDS, it is possible to define a unique family of sub-Markov
kernels $(P_t)_{t \geq 0} : E \times \mathcal{E} \to [0,1]$ such that for
all $t \geq 0$ and $f \in C_0(E)$, 
\[ \hat{P}_t f(x) = \int f(y) P_t(x, dy). \]
\end{prop}

A very important ingredient in the theory is the space of trajectories of a
FD processes (FD semigroup) as a probability space.  This space does not
appear explicitly in the study of labelled Markov processes but one does
see it in the study of continuous-time Markov chains and jump processes.
\begin{definition}
We define a \emph{trajectory} $\omega$ on $E_{\partial}$ to be a
\emph{cadlag}\footnote{By \emph{cadlag} we mean right-continuous
  with left limits.} function from $[0,\infty)\to E_{\partial}$ such that
if either $\omega(t-) = \partial$ or $\omega(t)=\partial$ then $\forall u
\geq t, \omega(u) = \partial$.  We can extend $\omega$ to a map from
$[0,\infty]$ to $E_{\partial}$ by setting $\omega(\infty)=\partial$.
\end{definition}

It is possible to associate to such an FDS a \emph{canonical FD process}.
Let $\Omega$ be the set of trajectories $\omega : [0, \infty) \to
E_\partial$. 
\begin{definition}
The
\emph{canonical FD process} associated to the FDS $(\hat{P}_t)$ is
\[(\Omega, \mathcal{G}, (\cG)_{t \geq 0}, (X_t)_{0 \leq t \leq \infty}, (\mathbb{P}^x)_{x \in E_\partial})\]
where
\begin{itemize}
\item $X_t(\omega) = \omega (t)$
\item $\mathcal{G} = \sigma (X_s ~|~ 0 \leq s < \infty)$, $\cG_t = \sigma (X_s ~|~ 0 \leq s \leq t)$
\item given any probability measure $\mu$ on $E_\partial$, by the
  Kolmogorov extension theorem, there exists a
  unique probability measure $\mathbb{P}^\mu$ on $(\Omega, \mathcal{G})$
  such that for all
  $n \in \mathbb{N}, 0 \leq t_1 \leq t_2 \leq ... \leq t_n$ and
  $x_0, x_1, ..., x_n$ in $E_\partial$,
\[ \mathbb{P}^\mu (X_0 \in dx_0, X_{t_1} \in dx_1, ..., X_{t_n} \in dx_n) =
  \mu (dx_0) P_{t_1}^{+\partial}(x_0, dx_1)...P_{t_n -
    t_{n-1}}^{+\partial}(x_{n-1}, dx_n) \]
  where $P_t^{+\partial}$ is the Markov kernel extending the Markov kernel $P_t$ to $E_\partial$ by $P_t^{+\partial} (x, \{ \partial \}) = 1 - P_t (x, E)$ and $P_t^{+ \partial} (\partial, \{ \partial \}) = 1$.
We set $\mathbb{P}^x = \mathbb{P}^{\delta_x}$.
\end{itemize}
\end{definition}
This is the version of the system that will be most useful for us. In order
to bring it more in line with the kind of transition systems that have
hitherto been studied in the computer science literature we introduce a
finite set of atomic propositions $AP$ and such a FD process is equipped
with a function $obs : E \to 2^{AP}$.  This function is extended to a
function $obs : E_\partial \to 2^{AP} \uplus \{ \partial \}$ by setting
$obs (\partial) = \partial$.

Instead of following the dynamics of the
system step by step as one does in a discrete system we have to study the
behaviour of sets of trajectories.  The crucial ingredient is the
distribution $\mathbb{P}^x$ which gives a measure on the space of
trajectories for a system started at the point $x$.

\subsection{Brownian motion as a FD process}

Brownian motion is a stochastic process describing the irregular motion of
a particle being buffeted by invisible molecules.  Now its range of
applicability extends far beyond its initial application~\cite{Karatzas12}.
The following definition is from~\cite{Karatzas12}.
\begin{defin}
A standard one-dimensional Brownian motion is a Markov process adapted to
the filtration $\cF_t$,
\[ B = (W_t,\cF_t), 0 \leq t < \infty  \]
defined on a probability space $(\Omega,\cF,P)$ with the properties
\begin{enumerate}
\item $W_0 = 0$ almost surely,
\item for $0\leq s < t$, $W_t - W_s$ is independent of $\cF_s$ and is
  normally distributed with mean $0$ and variance $t-s$.
\end{enumerate}
\end{defin}
In this very special process, one can start at any place, there is an
overall translation symmetry which makes calculations more tractable.
In order to do any calculations we use following fundamental formula:  If
the process is at $x$ at time $0$ then at time $t$ the probability that it
is in the (measurable) set $D$ is given by
\[ P_t(x,D) = \int_{y\in D} \frac{1}{\sqrt{2\pi t}}
  \exp\left(-\frac{(x-y)^2}{2t}\right)\mathrm{d}y. \]

\section{Bisimulation}
\label{sec:bisimulation}
The concept of bisimulation is fundamental and its history is well
documented~\cite{Sangiorgi09}.  We recall the definition of bisimulation on
continuous state spaces with discrete time
steps~\cite{Desharnais02,Panangaden09}, we call it a \emph{DT-bisimulation}
to emphasize that it pertains to discrete-time systems.  We consider LMPs
equipped with a family of atomic propositions $AP$ where $A\in AP$ is
interpreted on a specific LMP as a 
subset of the state space represented by its characteristic function $\chi_A$.
\begin{definition}
Given an LMP $(X, \Sigma, \tau, (\chi_P)_{P \in AP})$, a
\emph{DT-bisimulation} $R$ is an equivalence relation on $X$ such that if
$x R y$, then 
\begin{itemize}
\item for all $A \in AP$, $\chi_A(x) = \chi_A(y)$
\item for all $R$-closed sets $B \in \Sigma$, $\tau(x, B) = \tau(y, B)$.
\end{itemize}
\end{definition}

\subsection{Naive approach}

The key idea of bisimulation is that ``what can be observed now is the same''
and bisimulation is preserved by the evolution.  In order to capture this we
need two conditions: the first captures what is immediately observable and the
second captures the idea that the evolution preserves bisimulation.

Let us consider the naive extension of bisimulation in discrete time: let us consider an equivalence relation $R$ on the state space $E$ such that whenever $x~R~y$ ($x,y \in E$):
\begin{description}
\item[\namedlabel{init1}{(initiation 1)}] $obs (x) = obs(y)$, and
\item[\namedlabel{ind1}{(induction 1)}] for all $R$-closed sets $C$ in $\mathcal{E}$, for all time $t$, $P{t}(x,C) = P_{t}(y,C)$
\end{description}

Let us illustrate on an example why this definition is not enough.

We consider the case of Brownian motion on the reals where there is a single atomic proposition marking 0:
$obs (0) = 1$ and $obs(x) = 0$ for $x \neq 0$. Intuitively, we would like that two states $x$ and $y$ are bisimilar if and only if $|x| = |y|$ as the only symmetry that this system has is point reflection with respect to 0.

However, the two conditions \ref{init1} and \ref{ind1} are not strong enough to enforce that this equivalence relation is the greatest bisimulation.

Let us define the equivalence
\[ R = (\mathbb{R}^* \times \mathbb{R}^*) \cup \{ (0,0)\} \text{ where } \mathbb{R}^* = \mathbb{R}\setminus \{0\}. \]
This equivalence satisfies both conditions \ref{ind1} and \ref{init1}. The last one follows directly from the definitions of $R$ and $obs$.

For the induction condition, the only $R$-closed sets are $\emptyset, \{ 0\}, \mathbb{R}^*$ and $\mathbb{R}$, and for any state $z \neq 0$ and time $t \geq 0$, $P_t(z, \emptyset) = P_t (z, \{ 0\}) = 0$ and $P_t (z, \mathbb{R}^*) = P_t(z, \mathbb{R}) = 1$.

\subsection{Definition}
As we have just shown, unlike in
the discrete-time case we cannot just say that the ``next step'' preserves the
relation.  Therefore we have to talk about the
trajectories; but then we need to choose the right condition on sets of
trajectories. 
\begin{definition}
An equivalence relation $R$ on the state space $E$ is a \emph{bisimulation} if
whenever $x R y$, the following conditions are satisfied: 
\begin{description}
\item[\ref{init1}] $obs (x) = obs(y)$, and
\item[\namedlabel{ind2}{(induction 2)}]  for all $R$-closed sets $B$ in $\mathcal{G}$, $\mathbb{P}^x(B) =
  \mathbb{P}^y (B)$ where by $R$-closed, we mean that for all $\omega \in B$
  if a trajectory $\omega'$ is such that for all time $t \geq 0$, $\omega (t) R \omega' (t)
  $, then $\omega' \in B$. 
\end{description}
\end{definition}

Clearly equality is trivially a bisimulation. 
And by definition of $\mathbb{P}^x$, condition \ref{ind2} implies \ref{ind1}. 

We have chosen to give names to the conditions. The reason for choosing those names will become clear in section \ref{subsubsec:fork}.

\begin{remark}
Usually, for discrete time, instead of a single kernel $\tau$, a labelled Markov process is a family of Markov kernels indexed by a family of \emph{actions}. These actions correspond to the environment or the user acting on the process. The second condition of bisimulation is then stated on the corresponding Markov kernels for all actions. It is possible to to the same for continuous-time. We can consider a family of FD processes indexed by a set of actions. Condition \ref{ind2} is then stated for all these actions. Everything done afterwards can be adapted to that setting that way.
\end{remark}

\begin{lemma}
An equivalence relation $R$ is a bisimulation if and only if
whenever $x R y$, the following conditions are satisfied: 
\begin{description}
\item[\namedlabel{init2}{(initiation 2)}] for all $obs$-closed sets $B$ in $\mathcal{G}$, $\mathbb{P}^x(B) =
  \mathbb{P}^y (B)$ where by $obs$-closed, we mean that for all $\omega \in B$
  if a trajectory $\omega'$ is such that $obs \circ \omega = obs \circ \omega'
  $, then $\omega' \in B$. 
\item[\ref{ind2}] for all $R$-closed sets $B$ in $\mathcal{G}$, $\mathbb{P}^x(B) =
  \mathbb{P}^y (B)$.
\end{description}
\end{lemma}

\begin{proof}
Let us consider a bisimulation $R$. Let us now consider two states $x,y$ such that $x ~R~y$ and an $obs$-closed measurable set $B$. First note that the set $B$ is $R$-closed: if $\omega \in B$ and for all $t \geq 0$, $\omega (t) ~R~ \omega ' (t)$, then, by definition of bisimulation (initiation condition), $obs (\omega (t)) = obs (\omega ' (t))$ for all $t \geq 0$. Since the set $B$ is $obs$-closed, this means that $\omega ' \in B$ and hence $B$ is $R$-closed. Using the induction condition, we have that $\mathbb{P}^x (B) = \mathbb{P}^y (B)$.

Let us now consider an equivalence $R$ that satisfies both conditions. Let $x,y$ be two states such that $x ~R~ y$ and let us define the set $B_x = \{ \omega ~|~ obs (\omega (0)) = obs(x) \}$. The set $B_x$ is $obs$-closed and $\mathbb{P}^x (B_x) = 1$. Therefore $\mathbb{P}^y(B_x) = 1$ (by \ref{init2}) and therefore $obs(y) = obs(x)$.
\end{proof}

\begin{definition}
Two states are \emph{bisimilar} if there is a bisimulation that relates them.
\end{definition}

\begin{prop}
\label{prop:transitive-union-bisimulation}
Given two bisimulations $R_1$ and $R_2$, the transitive closure $R$ of $R_1
\cup R_2$ is a bisimulation. 
\end{prop}

\begin{proof}
Clearly $R$ is an equivalence.

Let us prove that the equivalence $R$ satisfies both conditions. Assume $x R
y$. This means that there is a finite sequence
\[x ~ R_1~ x_0~ R_2~ x_1~ R_1~  ... ~R_2~ x_n~ R_1~ y.\]

Let us consider an $obs$-closed set $B$ in $\mathcal{G}$, then since both $R_1$
and $R_2$ are bisimulations, we have that 
\[ obs(x) = obs(x_0) = obs(x_1) = ... = obs(x_n) = obs(y).\]

Let us now consider an $R$-closed set $B$. First, note that the set $B$ is $R_1$-closed: consider $\omega \in B$, and a trajectory $\omega'$ such that for all time $t \geq 0$, $\omega (t)~ R_1~ \omega ' (t)$. Then, in particular, for all time $t \geq 0$, $\omega (t)~ R~ \omega ' (t)$, and since $B$ is $R$-closed, $\omega' \in B$. Similarly, the set $B$ is $R_2$-closed.

Since $R_1$ is a bisimulation, we have that $\mathbb{P}^x(B) = \mathbb{P}^{x_0}(B)$, $\mathbb{P}^{x_n}(B) = \mathbb{P}^{y}(B)$ and $\mathbb{P}^{x_{2k+1}}(B) = \mathbb{P}^{x_{2k+2}}(B)$ (for all suitable $k$). And since $R_2$ is a bisimulation, we have that $\mathbb{P}^{x_{2k}}(B) = \mathbb{P}^{x_{2k}}(B)$ (for all suitable $k$). We then have:
\[ \mathbb{P}^x(B) = \mathbb{P}^{x_0}(B) = \mathbb{P}^{x_1}(B) = ... = \mathbb{P}^{x_n}(B) = \mathbb{P}^y(B)\]
\end{proof}

\begin{prop}
The relation ``is bisimilar to'' is the greatest bisimulation.
\end{prop}

\begin{proof}
Let us denote $R_{\max}$ the relation ``is bisimilar to'':
\[ R_{max} = \bigcup_{R \text{ bisimulation}} R \]

It is enough to prove that it is a bisimulation. First note that it is an
equivalence. Indeed, it is reflexive and symmetric since the equality is a
bisimulation.  For transitivity, note that if $x~ R_{\max}~ y$ and $y ~R_{\max}~ z$,
then there are two bisimulations $R_1$ and $R_2$ such that $x ~R_1~ y$ and
$y ~ R_2~ z$. By proposition \ref{prop:transitive-union-bisimulation}, we know that
the transitive closure $R$ of $R_1 \cup R_2$ is a bisimulation and in particular
this means that $x ~R~ y$, $y~R ~z$ and therefore $x ~R~z$. Since $R$ is a
bisimulation, this means that $x$ is bisimilar to $z$, which proves the
transitivity of $R_{\max}$.

Consider $x$ and $y$ such that $x$ is bisimilar to $y$. This means that there is
a bisimulation $R$ such that $x ~ R ~ y$. The initiation condition for $R$ gives us that $obs(x) = obs(y)$,
which also corresponds to the initiation condition we want for $R_{\max}$. Consider now
$B$ an $R_{\max}$-closed set. The set $B$ is also $R$-closed: consider $\omega \in B$ and a trajectory $\omega'$ such that for all time $t \geq 0$, $\omega (t) ~R~ \omega'(t)$. Then, we also have that for all time $t \geq 0$, $\omega (t)~ R_{max}~ \omega'(t)$ and since the set $B$ is $R_{max}$-closed, we have that $\omega' \in B$. And since $x~R~y$, we have that $\mathbb{P}^x(B) = \mathbb{P}^y(B)$ which concludes the proof.
\end{proof}

We now consider several examples and give their greatest bisimulation. Proving that an equivalence is the greatest bisimulation follows the following outline: first proving that the equivalence satisfies conditions \ref{init1} and \ref{ind2} (and hence it is a bisimulation), and then using \ref{init2} to prove that it is the greatest bisimulation possible.

\subsection{Basic examples}

\subsubsection{Deterministic Drift}

Consider a deterministic drift on the real line $\mathbb{R}$ with constant speed $a \in \mathbb{R}$. We consider two cases: with $0$ as the only
distinguished point and with all the integers distinguished from the other
points.

\paragraph*{With zero distinguished: }
Let us consider the case when there is a single atomic proposition called $obs$,
and $obs(x) = 1$ if and only if $x = 0$.

\begin{prop}
Two states $x$ and $y$ are bisimilar if and only if either $ax >0$ and $ay > 0$ or $x = y$.
\end{prop}

\begin{proof}
To make this proof not too tedious, we will assume that $a > 0$ (case $a < 0$
works in a similar fashion and case $a =0$ is boring). Denote 
\[ R = (\mathbb{R}_{>0} \times \mathbb{R}_{>0}) \cup \{ (x,x) ~|~ x \in \mathbb{R}_{\leq 0} \} \]
Let $x~R~y$.

%let us show that $P_t (x,C) = P_t(y,C)$. First
%note that for all $z \in \mathbb{R}$, $P_t(z, C) = \delta_C (z + at)$. There are
%several cases to consider: 
%\begin{itemize}
%\item If $x + at > 0$ and $C \cap \mathbb{R}^*_+ \neq \emptyset$, then since $C$
%  is $R$-closed, $\mathbb{R}^*_+ \subset C$ and therefore $x + at \in C$. It is
%  enough to prove that $y + at > 0$ to conclude in this case. There are two
%  subcases to consider. 
%\begin{itemize}
%\item First $x > 0$, in which case $y > 0$ since $xRy$. And finally $y+at >0$
%  since $a > 0$. 
%\item Second $x \leq 0$, in which case $x = y$ since $x Ry$. Therefore $y+at =
%  x+at > 0$. 
%\end{itemize}
%\item If $x + at > 0$ and $C \cap \mathbb{R}^*_+ = \emptyset$, then $x + at
%  \notin C$. As proven in previous case, we have that $y + at >0$. This means
%  that $y +at \notin C$ either. 
%\item if $x + at \leq 0$, then $x \leq 0$ and therefore $x = y$ since $x R
%  y$. 
%\end{itemize}

%Consider now an $obs$-closed measurable set $B$. We want to show that $\mathbb{P}^x (B)
%= \mathbb{P}^y(B)$. For that, there are two cases to consider: 
%\begin{itemize}
%\item either $x \leq 0$, in which case $x = y$ since $xRy$ which proves the point,
%\item or $x >0$, in which case $y > 0$. In that case, for all $t$, $\omega_x(t) >0$ and $\omega_y(t) > 0$ and in particular $obs \circ \omega_x = obs \circ \omega_y$ (it is always 0 in fact). Since $B$ is $obs$-closed, $\omega_x \in B$ if and only if $\omega_y \in B$.
%\end{itemize}

Let us consider $x$ and $y$ such that $x R y$. We have to consider two cases:
\begin{itemize}
\item If $obs(x) = 1$, this means that $x = 0$. The state $0$ is only bisimilar to itself, which means that $y = 0$ and therefore $obs(x) = obs(y)$.
\item If $obs(x) = 0$, this means that $x \neq 0$. The state $0$ is only bisimilar to itself, which means that $y \neq 0$ and therefore $obs(x) = obs(y)$.
\end{itemize} 

Consider a measurable set $B$. First, for any $z \in \mathbb{R}$, let us
denote $\omega_z$ the trajectory $\omega_z (t) = z + at$ and note that
$\mathbb{P}^z (B) = \delta_B (\omega_z)$.

Consider an $R$-closed measurable set $B$. We want to show that $\mathbb{P}^x (B)
= \mathbb{P}^y(B)$. For that, there are two cases to consider: 
\begin{itemize}
\item either $x \leq 0$, in which case $x = y$ since $xRy$ which proves the point,
\item or $x >0$, in which case $y > 0$. In that case, for all $t$, $\omega_x(t) >0$ and $\omega_y(t) > 0$ and in particular for all time $t$, $\omega_x (t) ~R~ \omega_y(t)$. Since $B$ is $R$-closed, $\omega_x \in B$ if, and only if, $\omega_y \in B$.
\end{itemize}

This concludes the second part of the proof.

Let us now prove that this is the greatest such bisimulation. We are using condition \ref{init2} for that:

Consider $x > 0$ and $y \leq 0$. For all $t \geq 0$, $\omega_x(t) \neq 0$, but $\omega_y(-y / a) = 0$. Define $B = \{ \omega ~|~ \omega(-y / a) = 0\}$. This set is $obs$-closed but $\mathbb{P}^x (B) = 0$ and $\mathbb{P}^y (B) = 1$. These two states cannot be bisimilar.

Consider $x,y \leq 0$. Note that $\omega_x(t) \neq 0$ for all $t \neq -x / a$. Define $B = \{ \omega ~|~ \omega(-y / a) = 0\}$. We have that $\mathbb{P}^y (B) = 1$ and the only way $\mathbb{P}^x (B) = 1$ is to have $-x / a = -y / a$, i.e. $x = y$. This concludes the proof.
\end{proof}

\paragraph*{With all integers distinguished: }
Let us consider the case when there is a single atomic proposition and $obs (x) = 1$ if and only if $x \in \mathbb{Z}$.

\begin{prop}
Two states $x$ and $y$ are bisimilar if and only if $x - \lfloor x \rfloor = y - \lfloor y \rfloor$, i.e. $x - y \in \mathbb{Z}$.
\end{prop}

\begin{proof}
Define $R = \{ (x,y) ~|~ x - y \in \mathbb{Z}\}$.

First let us prove that $R$ is a bisimulation. Take $x~ R~y$ and denote $k = x-y \in \mathbb{Z}$.

Note that $x \in \mathbb{Z}$ if and only if $y \in \mathbb{Z}$ and therefore $obs(x) = obs(y)$.

Consider an $R$-closed set $B$. In particular, this means that $k + B = \{ t \mapsto k+\omega(t) ~|~ \omega \in B\} = B$. Deterministic drift is invariant under translation, which means that $\mathbb{P}^y (B) = \mathbb{P}^{y +k}(B+k) = \mathbb{P}^x (B)$.

Let us now prove that this is the greatest bisimulation. Let $x,y \in \mathbb{R}$. Here we are going to assume that $a>0$, the case $a<0$ works exactly the same but considering $\lfloor x \rfloor$ instead of $\lceil x \rceil$. Define $z = \lceil x \rceil - x$. For any $s \in \mathbb{R}$, let us denote $\omega_s$ the trajectory $\omega_s (t) = s + at$.
Note that $\omega_x (z / a) = x + z = \lceil x \rceil \in \mathbb{Z}$ and $\omega_y (z / a) = y + z = y - x+ \lceil x \rceil$. This means that $\omega_y (z / a) \in \mathbb{Z}$ if and only if $y - x \in \mathbb{Z}$. Finally, define $B = \{ \omega ~|~ \omega (z / a) \in \mathbb{Z}\}$. This set is $obs$-closed and measurable, but we have proven that $\mathbb{P}^x(B) = \mathbb{P}^y(B)$ if and only if $y-x \in \mathbb{Z}$.
\end{proof}

\subsubsection{Fork}
\label{subsubsec:fork}

One could think that since trajectories are already included in the initiation condition \ref{init2}, the additional induction condition is not necessary. However, this example illustrates the crucial role of the induction condition in the definition of bisimulation. It is an extension of the standard ``vending machine'' example in discrete time to our continuous-time setting and it shows that even the condition \ref{ind1} are enough to discriminate between states that \ref{init2} cannot distinguish.

Let us consider the following state space:\\
\includegraphics[scale=0.3]{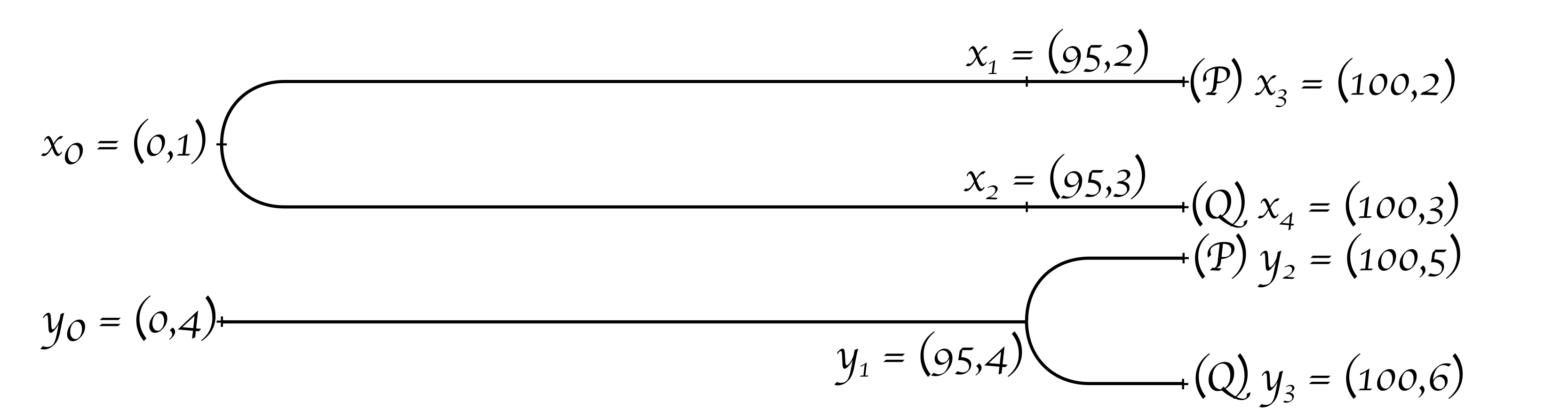}\\
There are two atomic propositions (denoted $P$ and $Q$ on the diagram), that are satisfied by the final state of some of the branches. The process is a drift at a constant speed to the right. When it reaches a fork, it moves to either branch with probability $1/2$ (and stops when he hits an atomic proposition).

The state space is made explicit as:
\[ \{(0,1)\} \uplus \left( (0,100] \times \{2,3\} \right) \uplus \left( [0,95] \times \{ 4\}\right) \uplus \left( (95,100] \times \{5,6\} \right) \]

We are going to use the following states:
\begin{align*}
x_0 & = (0,1) & y_0 & = (0,4)\\
x_1 & = (95,2) & y_1 & = (95,4)\\
x_2 & = (95,3)\\
x_3 & = (100,2) & y_2 & = (100,5)\\
x_4 & = (100,3) & y_3 & = (100,6)
\end{align*}
There are two atomic propositions $P$ and $Q$ and $obs(x_3) = obs(y_2) = (1,0)$, $obs(x_4) = obs(y_3) = (0,1)$ and $obs(z) = (0,0)$ otherwise.

The kernel is defined as follows for $t \leq 100$:
\begin{align*}
P_t(x_0, \{ (t,j) \}) & = \frac{1}{2} \quad \text{for } j = 2,3, ~ t \neq 0\\
P_t((x,j), (x+t, j)) & = 1 \quad \text{for all $j$ and for all $t$ such that $(x+t,j)$ exists}\\
P_t((y,4), (y+t, j)) & = \frac{1}{2} \quad \text{for $j = 5,6$ and for all $t$ such that $(y+t,j)$ exists}\\
P_t((100, j), (100, j)) & = 1 \quad \text{for } j = 2,3,5,6
\end{align*}

The basic claim is that the states $x_0$ and $y_0$ cannot be bisimilar since states $x_1, x_2, y_1$ cannot be bisimilar either. This is where condition \ref{ind2} is really important since the two states $x_0$ and $y_0$ have similar traces as they both satisfy the condition \ref{init2}.

\begin{prop}
The two states $x_0$ and $y_0$ satisfy the condition \ref{init2}.
\end{prop}

\begin{proof}
From state $x_0$, there are only two trajectories possible, each with probability $1/2$:
\[ \omega_1^x (t) = \begin{cases}
x_0 ~~\text{ if } t = 0\\
(t,2) ~~ \text{ otherwise}
\end{cases}
\qquad
\omega_2^x (t) = \begin{cases}
x_0 ~~\text{ if } t = 0\\
(t,3) ~~ \text{ otherwise}
\end{cases} \]

From state $y_0$, there are only two trajectories possible, each with probability $1/2$:
\[ \omega_1^y (t) = \begin{cases}
(t,4) ~~\text{ if } t \leq 95\\
(t,5) ~~ \text{ otherwise}
\end{cases}
\qquad
\omega_2^y (t) = \begin{cases}
(t,4) ~~\text{ if } t \leq 95\\
(t,6) ~~ \text{ otherwise}
\end{cases} \]

However, for all time $t \geq 0$, $obs(\omega^x_1(t)) = obs(\omega^y_1(t))$ and $obs(\omega^x_2(t)) = obs(\omega^y_2(t))$, which means that if a set $B$ is $obs$-closed, then $\omega^x_1 \in B$ (resp. $\omega^x_2 \in B$) if and only if $\omega^y_1 \in B$ (resp. $\omega^y_2 \in B$).

Putting all this together, we get that for any $obs$-closed set $B$:
\begin{align*}
\mathbb{P}^x(B)
& = \frac{1}{2} \delta_B (\omega^x_1) + \frac{1}{2} \delta_B (\omega^x_2)\\
& = \frac{1}{2} \delta_B (\omega^y_1) + \frac{1}{2} \delta_B (\omega^y_2)\\
& = \mathbb{P}^y(B).
\end{align*}
\end{proof}

\begin{prop}
The states $x_0$ and $y_0$ cannot be bisimilar.
\end{prop}

\begin{proof}
First, the states $x_1, x_2$ and $y_1$ cannot be bisimilar. Indeed, consider the set of trajectories $B = \{ \omega ~|~ obs(\omega (5)) = (1,0)\}$. This set is $obs$-closed (and measurable) but we have that $\mathbb{P}^{x_1} (B) = 1$, $\mathbb{P}^{x_2}(B) = 0$ and $\mathbb{P}^{x_1} (B) = 1/2$.

Second, the states $x_1$, $x_2$ and $y_1$ can only be bisimilar to themselves. Indeed, we can consider the set $B' = \{ \omega ~|~ obs(\omega (5)) = (1,0) ~\text{or}~ (0,1)\}$. For all $z \neq x_1, x_2$ or $y_1$, $\mathbb{P}^z(B') = 0$, whereas for $z = x_1, x_2$ or $y_1$, $\mathbb{P}^z(B') = 1$.

Third, we can consider the set $\{ x_1\}$. We have just shown that for any bisimulation $R$, the set $\{ x_1\}$ is $R$-closed. But then $P_t(x_0, \{ x_1\}) = 1/2$ whereas $P_t(y_0, \{ x_1\}) = 0$.
\end{proof}

\subsection{Examples based on Brownian motion}
\label{subsec:example_BM}

\subsubsection{Standard Brownian Motion}

\paragraph*{With zero distinguished: }
Let us consider the case when there is a single atomic proposition and $obs (x) = 1$ if and only if $x = 0$.

\begin{prop}
Two states $x$ and $y$ are bisimilar if and only if $|x| = |y|$.
\end{prop}

\begin{proof}
First let us prove that $R = \{ (x,y) ~|~ |x| = |y| \}$ is a bisimulation.

Consider $x~R~y$, i.e. $|x| = |y|$. This means that $x = 0$ if and only if $y = 0$. In other terms, $obs(x) = 1$ if and only if $obs(y) = 1$ and hence $obs(x) = obs(y)$.

Let now $B$ be an $R$-closed measurable set of trajectories. This means that $\omega \in B$ if and only if $- \omega \in B$ since $|\omega (t)| = |- \omega (t)|$ for all time $t \geq 0$. And therefore $\mathbb{P}^x(B) = \mathbb{P}^{-x}(-B) = \mathbb{P}^{-x}(B)$ where $-B := \{ t \mapsto - \omega(t) ~|~ \omega \in B \}$.

Let us now show that this is the greatest bisimulation. It can easily be seen that $0$ and $x \neq 0$ since $obs(0) \neq obs(x)$.

Let us now consider two different states $x$ and $y$. We can define the set $B_t = \{ \omega ~|~ \exists s < t ~ \omega (s) = 0\}$. This set is $obs$-closed. It can also be expressed as $B_t = T_0 ^{-1} ([0, t))$ where $T_0$ is the hitting time for Brownian motion and we know that for any state $z$,
\[ \mathbb{P}^z(B_t) = \sqrt{\frac{2}{\pi}} \int_{|z| \sqrt t} ^\infty e^{-s^2 / 2} ds \]
If $|x| \neq |y|$, it is impossible to have that $\mathbb{P}^x(B_t) = \mathbb{P}^y(B_t)$. This proves that no equivalence strictly bigger than $R$ may satisfy \ref{init2}.
\end{proof}

\paragraph*{With all integers distinguished: }
Let us consider the case when there is a single atomic proposition and $obs (x) = 1$ if and only if $x \in \mathbb{Z}$.

\begin{prop}
Two states $x$ and $y$ are bisimilar if and only if $x - \lfloor x \rfloor = y - \lfloor y \rfloor$ or $\lceil y \rceil - y$ .
\end{prop}

\begin{proof}
First let us prove that $R = \{ (x,y) ~|~ x - \lfloor x \rfloor = y - \lfloor y \rfloor \text{ or } \lceil y \rceil - y\}$ is indeed a bisimulation. This relies on the invariance under translation and symmetry of the problem.

Let us consider $x~R~y$. There are two cases to consider:
\begin{itemize}
\item $x - \lfloor x \rfloor = y - \lfloor y \rfloor$. Let $k = \lfloor y \rfloor - \lfloor x \rfloor \in \mathbb{Z}$, note that $x + k = y$. This means that $x \in \mathbb{Z}$ if and only if $y \in \mathbb{Z}$ and therefore $obs(x) = obs(y)$. Now consider an $R$-closed measurable set $B$ of trajectories. Since it is $R$-closed, we have that $B+k := \{ t \mapsto \omega (t) + k ~|~ \omega \in B\} = B$, which means that
\[ \mathbb{P}^x (B) = \mathbb{P}^{x + k} (B+k) = \mathbb{P}^y (B) \]
\item $x - \lfloor x \rfloor = \lceil y \rceil - y$. Using previous case, we can assume that $x$ and $y$ are in $[0,1]$ and $x = 1 - y$. We have that $x \in \{ 0,1\}$ if and only if $y \in \{ 0,1\}$ and therefore $obs(x) = obs(y)$. Now consider an $R$-closed measurable set $B$ of trajectories. Since it is $R$-closed, we have that $1-B := \{ t \mapsto 1 - \omega (t) ~|~ \omega \in B\} = B$, which means that
\[ \mathbb{P}^x (B) = \mathbb{P}^{1-x} (1-B) = \mathbb{P}^y (B) \]
\end{itemize}

Let us show that it is the greatest such bisimulation. Consider $x \notin \mathbb{Z}$ and $y \in \mathbb{Z}$. We have that $obs(x) \neq obs(y)$ and therefore these two states cannot be bisimilar.

Let us now consider $x,y \notin \mathbb{Z}$ such that $x Ry$ and the sets $B_t = \{ \omega ~|~ \exists s \in [0,t) ~ \omega (s) \in \mathbb{Z}\}$. These sets are $obs$-closed. Furthermore, they can also be expressed as:
\[B_t = \bigcup_{n \in \mathbb{N}} \{ \omega ~|~ T_n(\omega) <t\} = \bigcup_{n \in \mathbb{N}} T_n^{-1}([0,t)) \]
This proves that the sets $B_t$ are measurable. Let us compute $\mathbb{P}^z(B_t)$ for any $z \in \mathbb{R}$:
\begin{align*}
\mathbb{P}^z (B_t)
& = \mathbb{P}^z \left(\left(T_{\lfloor x \rfloor} \wedge T_{\lceil x \rceil}\right)^{-1} \left([0,t)\right)\right)\\
& = \mathbb{P}^{z- \lfloor z \rfloor} \left(\left(T_0 \wedge T_1\right)^{-1} \left([0,t)\right)\right)\\
& = \int_0^t \mathbb{P}^{z- \lfloor z \rfloor} \left(\left(T_0 \wedge T_1\right) \in ds \right)
\end{align*}
Since this is true for all $t \geq 0$ and all $z \in \mathbb{R}$, we have that $\mathbb{P}^{x- \lfloor x \rfloor} \left(\left(T_0 \wedge T_1\right) \in ds \right) = \mathbb{P}^{y- \lfloor y \rfloor} \left(\left(T_0 \wedge T_1\right) \in ds \right)$ and therefore the following Laplace transforms are equal:
\[ \mathbb{E}^{x- \lfloor x \rfloor} [\exp (- \lambda (T_0 \wedge T_1))] = \mathbb{E}^{y- \lfloor y \rfloor} [\exp (- \lambda (T_0 \wedge T_1))] \]
Using~\cite{Karatzas12}, we have that
\[ \mathbb{E}^{z} [\exp (- \lambda (T_0 \wedge T_1))] = \frac{\cosh \left( \left( z - \frac{1}{2} \right) \sqrt{2 \lambda} \right)}{\cosh \left( \frac{1}{2}\sqrt{2 \lambda} \right)} \]
Therefore, we know that
\[ \cosh \left(  \left( x- \lfloor x \rfloor - \frac{1}{2} \right) \sqrt{2 \lambda} \right) = \cosh \left(  \left( y- \lfloor y \rfloor - \frac{1}{2} \right) \sqrt{2 \lambda} \right) \]
Using simple properties of $\cosh$, we get that either $x- \lfloor x \rfloor - \frac{1}{2} = y- \lfloor y \rfloor - \frac{1}{2}$ (i.e. $x- \lfloor x \rfloor = y- \lfloor y \rfloor$) or $x- \lfloor x \rfloor - \frac{1}{2} = \frac{1}{2} + \lfloor y \rfloor - y$ (i.e. $x- \lfloor x \rfloor = 1 + \lfloor y \rfloor - y = \lceil y \rceil -y$). This proves that no equivalence strictly bigger than $R$ may satisfy \ref{init2}.
\end{proof}

\paragraph*{With an interval distinguished: }
Let us consider the case when there is a single atomic proposition and $obs (x) = 1$ if and only if $x \in [-1,1]$.

\begin{prop}
Two states $x$ and $y$ are bisimilar if and only if $|x| = |y|$.
\end{prop}

\begin{proof}
First let us prove that $R = \{ (x,y) ~|~ |x| = |y| \}$ is a bisimulation.

Consider $x ~R~ y$, i.e. $|x| = |y|$. Clearly, $x \in [-1,1]$ if and only if $y \in [-1,1]$ and therefore $obs(x) = obs(y)$. Let us now look at the induction condition. Consider an $R$-closed measurable set $B$ of trajectories. This means that $-B := \{ t \mapsto - \omega (t) ~|~ \omega \in B\} = B$ and therefore $\mathbb{P}^x (B) = \mathbb{P}^y (B)$.

Let us now prove that this is the greatest bisimulation.

For $x \in [-1, 1]$ and $y \notin [-1,1]$, clearly $obs(x) \neq obs(y)$ which means that $x$ and $y$ cannot be bisimilar.

Let $x,y \notin [-1, 1]$. Let us define the sets $B_t = \{ \omega ~|~ \exists s \in [0,t) ~ \omega (s) \in [-1,1]\}$. This set is $obs$-closed, however, for $z > 1$, we have that
\[ \mathbb{P}^z (B_t) = \mathbb{P}^{|z|} (B_t) = \mathbb{P}^z (T_1^{-1} ([0,t))) = \sqrt{\frac{2}{\pi}} \int_{(1 - |z|) \sqrt{t}}^\infty e^{-s^2 / 2}ds \]
Therefore if there is a bisimulation $R$ such that $x R y$, in particular we have that $\mathbb{P}^x (B_t) = \mathbb{P}^y (B_t)$ and hence $|x| = |y|$.

Let $x,y \in [-1, 1]$ such that $x Ry$. Define $B_t = \{ \omega ~|~ \exists s \in [0,t) ~ \omega (s) \notin [-1,1]\}$. Note that
\[ \mathbb{P}^z(B_t) = \mathbb{P}^{z+1} (T_0 \wedge T_2 < t) = \int_0^t \mathbb{P}^{z+1} (T_0 \wedge T_2 \in ds)\]
Since we have that $\mathbb{P}^x(B_t) = \mathbb{P}^y(B_t)$ for all $t \geq 0$, we get that $\mathbb{P}^{x+1} (T_0 \wedge T_2 \in ds) = \mathbb{P}^{y+1} (T_0 \wedge T_2 \in ds)$ and therefore the corresponding Laplace transforms are equal:
\[ \mathbb{E}^{x+1} [\exp (- \lambda (T_0 \wedge T_2))] = \mathbb{E}^{y+1} [\exp (- \lambda (T_0 \wedge T_2))] \]
Using~\cite{Karatzas12}, we have that
\[ \mathbb{E}^{z} [\exp (- \lambda (T_0 \wedge T_2))] = \frac{\cosh \left( \left( z - 1 \right) \sqrt{2 \lambda} \right)}{\cosh \left( \sqrt{2 \lambda} \right)} \]
Therefore, we know that
\[ \cosh \left(  x \sqrt{2 \lambda} \right) = \cosh \left(  y \sqrt{2 \lambda} \right) \]
Using simple properties of $\cosh$, we get that either $x = y$ or $x = -y$ which concludes the proof.
\end{proof}

\subsubsection{Brownian motion with drift}

Let us consider a Brownian process with drift: $W'_t = W_t + at$ (where $W_t$ is the standard Brownian motion and $a > 0$, note that the case $a<0$ is symmetric).

\paragraph*{With zero distinguished: }
Let us consider the case when there is a single atomic proposition and $obs (x) = 1$ if and only if $x = 0$.

\begin{prop}
Two states $x$ and $y$ are bisimilar if and only if $x = y$.
\end{prop}

\begin{proof}
As stated before, the equivalence where a state is only related to itself is a bisimlation.

Let us now show that this is the greatest bisimulation. Let us consider two different states $x$ and $y$. Similarly to what we did for the standard Brownian motion, we can rule out the case where $x = 0$ (and $y \neq 0$) or $y = 0$ (and $x \neq 0$) by simply looking at the function $obs$.

We can define the set $B_t = \{ \omega ~|~ \exists s < t ~ \omega (s) = 0\}$. This set is $obs$-closed. It can also be expressed as $B_t = T_0 ^{-1} ([0, t))$ where $T_0$ is the hitting time for Brownian motion and we know that for any state $z$,
\[ \mathbb{P}^z(T_0 \in ds) = \frac{|z|}{\sqrt{2 \pi s^3}} \exp \left(- \frac{(z + as)^2}{2s} \right) ds \]
Since we have that for all $t$, $\mathbb{P}^x (B_t) = \mathbb{P}^y (B_t)$, then we also have that for all $s \geq 0$,
\[ |x| \exp \left(- \frac{(x + as)^2}{2s} \right) = |y| \exp \left(- \frac{(y + as)^2}{2s} \right) \]
Since $x,y \neq 0$, we have that for all $s,t \geq 0$,
\[ -\frac{(y + as)^2}{2s} + \frac{(x + as)^2}{2s} = -\frac{(y + at)^2}{2t} + \frac{(x + at)^2}{2t} \]
which is equivalent to $(s-t)y^2 = (s-t) x^2$. This means that in that case $|x| = |y|$. Going back to the original expression, we have that for all $s \geq 0$, $-(x+as)^2 = -(y+as)^2$ and therefore $2asx = 2asy$. Since $a \neq 0$, we get that $x = y$ in order to have \ref{init2}.
\end{proof}

\paragraph*{With all integers distinguished: }
Let us consider the case when there is a single atomic proposition and $obs (x) = 1$ if and only if $x \in \mathbb{Z}$.

\begin{prop}
Two states $x$ and $y$ are bisimilar if and only if $x - \lfloor x \rfloor = y - \lfloor y \rfloor$.
\end{prop}

\begin{proof}
First let us prove that $R = \{ (x,y) ~|~ x - \lfloor x \rfloor = y - \lfloor y
\rfloor\}$ is indeed a bisimulation. This relies on the invariance under
translation of the problem. Note that compared to standard Brownian motion with all integers distinguished, the drift ``removes'' the invariance under symmetry.

Indeed, let us consider $x~R~y$. Let $k = y - x \in \mathbb{Z}$ (i.e. $x + k = y$). Clearly $obs(x) = obs(y)$ as in the standard Brownian motion case.
Let us consider an $R$-closed set $B$ ($B + k = B$ as in the standard case). We have that:
\[ \mathbb{P}^x (B) = \mathbb{P}^{x + k} (B+k) = \mathbb{P}^y (B) \]

Let us show that it is the greatest such bisimulation. Similarly to the standard case, $x \notin \mathbb{Z}$ and $y \in \mathbb{Z}$ cannot be bisimilar since they don't have the same observables.

Let us now consider $x,y \notin \mathbb{Z}$ and the sets $B_t = \{ \omega ~|~ \exists s \in [0,t) ~ \omega (s) \in \mathbb{Z}\}$ such that for all $t$, $\mathbb{P}^x(B_t) = \mathbb{P}^y(B_t) $. Similarly to what we did in the case of standard BM, for all $z$,
\[ \mathbb{P}^z (B_t) = \mathbb{P}^{z - \lfloor z \rfloor} (T_0 \wedge T_1 < t) \]
Since for all $t$, $\mathbb{P}^x(B_t) = \mathbb{P}^y(B_t) $, we get that $\mathbb{E}^{x - \lfloor x\rfloor} [e ^{-\lambda (T_0 \wedge T_1)}] = \mathbb{E}^{y - \lfloor y\rfloor} [e ^{-\lambda (T_0 \wedge T_1)}]$. For $0 \leq z < 1$ and all $\lambda \geq 0$,
\begin{align*}
\mathbb{E}^z [e ^{-\lambda (T_0 \wedge T_1)}]
& = \mathbb{E}^z [e ^{-\lambda T_0} ~|~ T_0 < T_1] + \mathbb{E}^z [e ^{-\lambda T_1} ~|~ T_1 < T_0]\\
& = \mathbb{E}^z \left[e ^{-\lambda T_0} ~\left|~ \sup_{0 \leq s \leq T_0} W^{(a)}_s < 1 \right. \right] + \mathbb{E}^z \left[e ^{-\lambda T_1} ~\left|~ \inf_{0 \leq s \leq T_1} W^{(a)}_s >0\right.\right]\\
& = \frac{\sinh ((1-z) \sqrt{2 \lambda + a^2}) e^{-a z} + \sinh(z \sqrt{2 \lambda + a^2}) e ^{a (1-z)}}{\sinh( \sqrt{2 \lambda + a^2})}
\end{align*}

We can denote $k = \sqrt{2 \lambda + a^2}$ and we can define for $z \in (0,1)$ and $k \geq a$,
$g_z(k) = \sinh ((1-z) k) e^{-a z} + \sinh(z k) e ^{a (1-z)}$. We have that for
all $k \geq a$, $g_{x - \lfloor x\rfloor}(k) = g_{y - \lfloor y\rfloor}(k)$. We
want to prove that $x - \lfloor x\rfloor = y - \lfloor y\rfloor$. This is done
through the following lemma. 

\begin{lemma}
Consider $z_1, z_2 \in (0,1)$. If $g_{z_1}(k) = g_{z_2}(k)$ for all $k \geq a$, then $z_1 = z_2$.
\end{lemma}

\begin{proof}{Of lemma.}
First, note that for $z \in (0,1)$,
\begin{align*}
\frac{\ln g_z(k)}{k}
& \sim_{k \to \infty} \frac{1}{k} \ln \left( \frac{e^{k(1-z)}}{2 e^{az}} + \frac{e^{kz}}{2 e^{a(z-1)}} \right)\\
& \sim_{k \to \infty} \frac{1}{k} \ln \left( e^{k\max \{ (1-z), z\}} \right) = \max \{ 1-z, z\}
\end{align*}
Since for all $k \geq a$, $g_{z_1}(k) = g_{z_2}(k)$, we get that $\max \{ 1- z_1, z_1\} = \max \{ 1- z_2, z_2\}$, i.e. $z_1 = z_2$ or $z_1 = 1-z_2$.

If $z_1 = 1/2$, then both cases are $z _1 = z_2 = 1/2$.

Let us study the second case $ z_1 = z$ and $z_2 = 1-z$ for $z \in (0,1) \setminus \{ 1/2\}$. We have that for all $k \geq a$, $g_{z}(k) = g_{1-z}(k)$. This equation amounts to
\[ \sinh (k (1-z)) e^{-az} + \sinh (kz) e^{a (1-z)} = \sinh (kz) e^{-a(1-z)} + \sinh(k(1-z)) e^{az} \]
By reorganizing the terms, we get that
\[ \frac{\sinh (k (1-z))}{\sinh (kz)} = \frac{\sinh (a (1-z))}{\sinh (az)} \]
Considering the left hand-side, we have that:
\[\frac{\sinh (k (1-z))}{\sinh (kz)} \sim_{k \to \infty} \frac{e^{k (1-z)}}{e^{kz}} = e ^{k (1- 2z)}\]
We can therefore consider the limit of the left-hand side of the equation:
\[ \lim_{k \to +\infty} \frac{\sinh (k (1-z))}{\sinh (kz)} = \begin{cases}
+ \infty \quad \text{if } z < 1/2\\
0 \quad \text{if } z > 1/2
\end{cases} \]
However, that would mean
\[ \frac{\sinh (a (1-z))}{\sinh (az)} = \begin{cases}
+ \infty \quad \text{if } z < 1/2\\
0 \quad \text{if } z > 1/2
\end{cases} \]
which is impossible. Therefore we get that $z_1 = z_2$.
\end{proof}
This means that in order to satisfy \ref{init2}, we need to have that $x - \lfloor x\rfloor = y - \lfloor y\rfloor$ and therefore $R$ is the greatest bisimulation.
\end{proof}

\paragraph*{With an interval distinguished: }
Let us consider the case when there is a single atomic proposition and $obs (x) = 1$ if and only if $x \in [-1,1]$.

\begin{prop}
Two states $x$ and $y$ are bisimilar if and only if $x = y$.
\end{prop}

\begin{proof}
As stated, the equality is a bisimulation. Let us prove that it is the greatest.

Clearly, $x \in [-1, 1]$ and $y \notin [-1,1]$ don't have the same observables, which means that $x$ and $y$ cannot be bisimilar.

Let $x,y \notin [-1, 1]$ and for all $t \geq 0$, $B_t = \{ \omega ~|~ \exists s < t ~ \omega (s) \in [-1,1]\}$ such that for all $t \geq 0$, $\mathbb{P}^x (B_t) = \mathbb{P}^y (B_t)$. Then, for all $\lambda \geq 0$, $\mathbb{E}^x [e^{- \lambda (T_{-1} \wedge T_1)}] = \mathbb{E}^y [e^{- \lambda (T_{-1} \wedge T_1)}]$.

For $z > 1$, we have that
\begin{align*}
\mathbb{E}^z [e^{- \lambda (T_{-1} \wedge T_1)}]
& = \mathbb{E}^z [e^{- \lambda T_1}]\\
& = \exp \left( a (1-z) - |1-z| \sqrt{2 \lambda + a^2} \right)\\
& = \exp \left( a (1-z) + (1-z) \sqrt{2 \lambda + a^2} \right)\\
& = \exp \left( (1-z) (a+ \sqrt{2 \lambda + a^2}) \right)
\end{align*}
This function is injective on $[1, +\infty)$ which means that we cannot have both $x,y >1$.

For $z < -1$, we have that
\begin{align*}
\mathbb{E}^z [e^{- \lambda (T_{-1} \wedge T_1)}]
& = \mathbb{E}^z [e^{- \lambda T_{-1}}]\\
& = \exp \left( a (-1-z) - |-1-z| \sqrt{2 \lambda + a^2} \right)\\
& = \exp \left( -a (1+z) + (1+z) \sqrt{2 \lambda + a^2} \right)\\
& = \exp \left( (1+z) (-a+ \sqrt{2 \lambda + a^2}) \right)
\end{align*}
This function is injective on $(- \infty , 1]$ which means that we cannot have both $x,y <-1$.

Assume $x > 1$ and $y < -1$. We then have that
\[ \exp \left( (1-x) (a+ \sqrt{2 \lambda + a^2}) \right) = \exp \left( (1+y) (-a+ \sqrt{2 \lambda + a^2}) \right) \]
which means that $(1-x) (a+ \sqrt{2 \lambda + a^2}) = (1+y) (-a+ \sqrt{2 \lambda + a^2})$, i.e. $2a = (a+ \sqrt{2 \lambda + a^2}) x + (-a+ \sqrt{2 \lambda + a^2}) y$.
For $\lambda = 0$, we get that $x=1$, which is not possible, and therefore $x$ and $y$ cannot be bisimilar.
%For simplicity, we will assume that $x,y > 1$. Let us define the sets $B_t = \{ \omega ~|~ \exists s \in [0,t) ~ \omega (s) \in [-1,1]\}$. This set is $obs$-closed, however, for $z > 1$, we have that
%\[ \mathbb{P}^z (B) = \mathbb{P}^z (T_1^{-1} ([0,t))) = \int_0^t \frac{|1-z|}{\sqrt{2 \pi s^3}} \exp \left( - \frac{(1-z-as^3)^2}{2s}\right) ds\]
%\textbf{TODO: conclude ?}

Let $x,y \in [-1, 1]$ such that for all $t \geq 0$, $\mathbb{P}^x(B_t) = \mathbb{P}^y(B_t)$ with $B_t = \{ \omega ~|~ \exists 0 \leq s < t ~ \omega (s) \notin [-1,1] \}$. As we did before, for all $\lambda \geq 0$, $\mathbb{E}^x [e^{-\lambda (T_{-1} \wedge T_1)}] = \mathbb{E}^y [e^{-\lambda (T_{-1} \wedge T_1)}]$.
For all $z \in [-1,1]$, we have that
\begin{align*}
\mathbb{E}^z [e ^{- \lambda (T_{-1} \wedge T_1)}]
& = \mathbb{E}^z [e ^{- \lambda T_{-1}} ~|~ T_{-1} < T_1] + \mathbb{E}^z [e ^{- \lambda T_1} ~|~ T_1 < T_{-1}]\\
& = \mathbb{E}^z [e ^{- \lambda T_{-1}} ~|~ \sup_{0 \leq s \leq T_{-1}} W_s < 1] + \mathbb{E}^z [e ^{- \lambda T_1} ~|~ \inf_{0 \leq s \leq T_{-1}} W_s > -1]\\
& = \frac{\sinh ((z+1) \sqrt{2 \lambda + a^2}) e ^{a(1-z)} + \sinh ((1-z) \sqrt{2 \lambda + a^2}) e ^{-a(1+z)}}{\sinh (2 \sqrt{2 \lambda + a^2})}
\end{align*}

We can denote $k = \sqrt{2 \lambda + a^2}$ and we can define for $z \in (0,1)$ and $k \geq a$,
$h_z(k) = \sinh ((z+1) k) e ^{a(1-z)} + \sinh ((1-z) k) e ^{-a(1+z)}$. We have that for all $k \geq a$, $h_{x}(k) = h_{y}(k)$. We want to prove that $x = y$. This is done through the following lemma.

\begin{lemma}
Consider $z_1, z_2 \in [-1,1]$. If $h_{z_1}(k) = h_{z_2}(k)$ for all $k \geq a$, then $z_1 = z_2$.
\end{lemma}

\begin{proof}{Of lemma.}
First, note that for $z \in [-1,1]$,
\begin{align*}
\frac{\ln h_z(k)}{k}
& \sim_{k \to \infty} \frac{1}{k} \ln \left( \frac{e^{k(1+z)}}{2 e^{-a(1-z)}} + \frac{e^{k(1-z)}}{2 e^{a(z+1)}} \right)\\
& \sim_{k \to \infty} \frac{1}{k} \ln \left( e^{k\max \{ (1-z), 1+z\}} \right) = \max \{ 1-z, 1+z\}
\end{align*}
Since for all $k \geq a$, $g_{z_1}(k) = g_{z_2}(k)$, we get that $\max \{ 1- z_1, 1+z_1\} = \max \{ 1- z_2, 1+z_2\}$, i.e. $z_1 = z_2$ or $z_1 = -z_2$.

If $z_1 = 0$, then both cases are $z _1 = z_2 = 0$.

Let us study the second case $ z_1 = z$ and $z_2 = -z$ for $z \in [-1,1] \setminus \{ 0\}$. We have that for all $k \geq a$, $h_{z}(k) = h_{-z}(k)$. This equation amounts to
\[ \sinh ((1+z) k) e ^{a(1-z)} + \sinh ((1-z) k) e ^{-a(1+z)} = \sinh ((1-z) k) e ^{a(1+z)} + \sinh ((1+z) k) e ^{-a(1-z)} \]
By reorganizing the terms, we get that
\[ \frac{\sinh (k (1+z))}{\sinh (k(1-z))} = \frac{\sinh (a (1+z))}{\sinh (a(1-z))} \]
Considering the left hand-side, we have that:
\[\frac{\sinh (k (1+z))}{\sinh (k(1-z))}
\sim_{k \to \infty} \frac{e^{k (1+z)}}{e^{k(1-z)}} = e ^{2kz}\]
We can therefore consider the limit of the left-hand side of the equation:
\[ \lim_{k \to +\infty} \frac{\sinh (k (1+z))}{\sinh (k(1-z))} = \begin{cases}
+ \infty \quad \text{if } z >0\\
0 \quad \text{if } z <0
\end{cases} \]
However, that would mean
\[ \frac{\sinh (a (1+z))}{\sinh (a(1-z))} = \begin{cases}
+ \infty \quad \text{if } z >0\\
0 \quad \text{if } z <0
\end{cases} \]
which is impossible. Therefore we get that $z_1 = z_2$.
\end{proof}
With this the overall proof is complete.
\end{proof}

\subsubsection{Brownian motion with absorbing wall}

Another usual variation on Brownian motion is to add boundaries and to consider that the process does not move anymore or dies once it has hit a boundary. Since all our previous examples involved probability distributions (as opposed to subprobabilities), we will see the boundary as killing the process.

\paragraph*{Absorption at 0: }
let us consider the case of Brownian motion with absorption at the origin
and without any atomic proposition. The state space is $\mathbb{R}_{>0}$.

\begin{prop}
Two states $x$ and $y$ are bisimilar if and only if $x = y$.
\end{prop}

\begin{proof}
We know that equality is a bisimulation. Let us prove that it is the greatest.

For all $t \geq 0$, the set $B_t = \{ \omega ~|~ obs \circ \omega (t)
= \partial \}$ is $obs$-closed. Let us clarify the intuition behind that set $B_t$: it is the set of trajectories such that the process following one of these trajectories is dead at time $t$.

For all state $x \geq 0$, 
\begin{align*}
\mathbb{P}_{abs}^x(B_t)
& = \mathbb{P}^x (T_0 < t)\\
& = \frac{2}{\sqrt{2 \pi}} \int_{x / \sqrt{t}}^{+ \infty} e^{-z^2 /2} dz
\end{align*}
The only way $\mathbb{P}_{abs}^x(B_t) = \mathbb{P}_{abs}^y(B_t)$ is therefore to have $x = y$.
\end{proof}

\paragraph*{Absorption at 0 and $b$: }
let us consider the case of Brownian motion with absorption at the origin and at $b > 0$ and without any atomic proposition. The state space is therefore $(0,b)$.

\begin{prop}
\label{prop:BM_double_absorption}
Two states $x$ and $y$ are bisimilar if and only if $x = y$ or $x = b - y$.
\end{prop}

\begin{proof}
Let us define the equivalence
\[ R = \{ (x,x), (x, b-x) ~|~ x \in (0,b)\} \]

First note that it is indeed a bisimulation. There are no atomic propositions. Let us now consider an $R$-closed set $B$ of trajectories. This means that $b - B := \{ t \mapsto b - \omega (t) ~|~ \omega \in B \} = B$. So for $x \in (0,b)$, we have that
\[ \mathbb{P}^x(B) = \mathbb{P}^{b - x} (b - B) = \mathbb{P}^{b - x} (B) \]
This means that $R$ is a bisimulation. 

Let us now prove that it is the greatest bisimulation. For all $t \geq 0$,
the set $B_t = \{ \omega ~|~ \omega (t) = \partial \}$ is $obs$-closed. For
all $x \in (0,b)$ and $t \geq 0$, 
\[ \mathbb{P}_{abs}^x(B_t) = \mathbb{P}^x (T_0 \wedge T_b < t) \] 
Similarly to what was done in the case of standard Brownian motion with all
integers distinguished, we get that $\mathbb{P}_{abs}^x(B_t) =
\mathbb{P}_{abs}^y(B_t)$ for all $t\geq 0$ if and only if $x = y$ or $x =
b-y$. 
\end{proof}

\paragraph*{Absorption at 0 and $2b$ with atomic proposition at $b$: }
let us consider the case of Brownian motion with absorption at the origin
and at $2b > 0$, so the state space is $(0,2b)$, and with a single atomic proposition such that $obs(b) = 1$
and $obs(x) =0$ for $x \neq b$. 

\begin{prop}
Two states $x$ and $y$ are bisimilar if and only if $x = y$ or $y = 2b - x$.
\end{prop}

\begin{proof}
Let us define the equivalence
\[ R = \{ (x,x), (x, 2b-x) ~|~ x \in (0,2b)\} \]

Let us show that this relation $R$ is a bisimulation. Clearly $x = b$, if and only if, $2b - x = b$ and therefore $obs(x) = obs(2b -x)$. The proof of the induction condition is similar to claim \ref{claim:BM_double_absorption}.

Similarly to proposition \ref{prop:BM_double_absorption}, we also have that this is the greatest bisimulation.
\end{proof}

\paragraph*{Absorption at 0 and $4b$ with atomic proposition at $b$: }
let us consider the case of Brownian motion with absorption at the origin
and at $4b > 0$, so the state space is $(0,4b)$, and with a single atomic proposition such that $obs(b) = 1$
and $obs(x) =0$ for $x \neq b$. 

\begin{prop}
Two states $x$ and $y$ are bisimilar if and only if $x = y$. 
\end{prop}

\begin{proof}
Using proposition \ref{prop:BM_double_absorption}, it is clear that a state $x$ can only be bisimilar to either itself or $4b - x$. First, state $b$ is not bisimilar to $3b$ since $obs(b) = 1$ and $obs(3b) = 0$.

Let us now show that $x \in (0, 2b)$ with $x \neq b$ and $4b-x \in (2b,
4b)$ are not bisimilar. Let us define $B_t = \{ \omega ~|~ \exists s \in
[0,t) ~ \omega(s) = b\}$. This set is indeed $obs$-closed. Similarly to
what was done in the standard Brownian motion case, we can compare
$\mathbb{E}^x_{abs}[e^{-\lambda T^{abs}_b}]$ and
$\mathbb{E}^{4b-x}_{abs}[e^{-\lambda T^{abs}_b}]$ instead. 

For $z <b$,
\begin{align*}
\mathbb{E}^z_{abs} [e^{-\lambda T^{abs}_b}]
& = \int_0 ^\infty e^{-\lambda t} \mathbb{P}^z_{abs} [T^{abs}_b \in dt]\\
& = \int_0 ^\infty e^{-\lambda t} \mathbb{P}^z [T_b \in dt, T_b < T_0]\\
& = \int_0 ^\infty e^{-\lambda t} \mathbb{P}^z [T_b \in dt ~|~ T_b < T_0]
  \mathbb{P}^z (T_b < T_0)\\ 
& = \mathbb{E}^z [e^{-\lambda T_b} ~|~ T_b < T_0] \mathbb{P}^z (T_b < T_0)\\
& = \frac{z \sinh (z \sqrt{2 \lambda})}{b \sinh (b \sqrt{2 \lambda})}
\end{align*}
Similarly, for $z > b$,
\begin{align*}
\mathbb{E}^z_{abs} [e^{-\lambda T^{abs}_b}]
& = \int_0 ^\infty e^{-\lambda t} \mathbb{P}^z_{abs} [T^{abs}_b \in dt]\\
& = \int_0 ^\infty e^{-\lambda t} \mathbb{P}^z [T_b \in dt, T_b < T_{4b}]\\
& = \int_0 ^\infty e^{-\lambda t} \mathbb{P}^z [T_b \in dt ~|~ T_b <
  T_{4b}] \mathbb{P}^z (T_b < T_{4b})\\ 
& = \mathbb{E}^z [e^{-\lambda T_b} ~|~ T_b < T_{4b}] \mathbb{P}^z (T_b < T_{4b})\\
& = \frac{(4b - z) \sinh ((4b - z) \sqrt{2 \lambda})}{3b \sinh (3b \sqrt{2 \lambda})}
\end{align*}
This function is strictly decreasing on $(b, 4b)$, hence we cannot have $x
\in (b, 2b)$ bisimilar to $4b - x$. 
Moreover for $x < b$, we get
\[ \mathbb{E}^{4b -x}_{abs} [e^{-\lambda T^{abs}_b}] = \frac{x \sinh (x
    \sqrt{2 \lambda})}{3b \sinh (3b \sqrt{2 \lambda})} \] 
And we get that $\mathbb{E}^x_{abs} [e^{-\lambda T^{abs}_b}] =
\mathbb{E}^{4b -x}_{abs} [e^{-\lambda T^{abs}_b}]$ if and only if $3b \sinh
(3b \sqrt{2 \lambda}) = b \sinh (b \sqrt{2 \lambda})$ (since $x \neq 0$)
which is not the case. 
\end{proof}

\section{Feller-Dynkin cospan}
\label{sec:FDcospan}
The concept of bisimulation that we have discussed so far is defined between
states of a process.  One often wants to compare different processes with
different state spaces.  For this one needs to use functions that relate the
state spaces of different processes.  One does want to preserve the relational
character of bisimulation.  In the coalgebra literature one uses spans of
so-called ``zigzag'' morphisms.  In previous work~\cite{Danos06} on (discrete-time) Markov
processes people have considered cospans as this leads to a smoother theory.
Intuitively, the difference is whether one thinks of an equivalence relation as
a set of ordered pairs or as a collection of equivalence classes.

\subsection{Feller-Dynkin homomorphism}

This definition of bisimulation can easily be adapted to states in different
Markov processes by constructing the disjoint union of the Markov processes. 

The disjoint union of two Markov processes is defined as such: given two FD processes
$(E_j, \mathcal{E}_j, (\hat{P}_t^j), (P_t^j), \Omega_j, \mathcal{G}_j, (\mathbb{P}_j^x), obs_j)_{j = 1,2}$,
we write $i_1 : E_1 \to E_1 \uplus E_2$ and $i_2 : E_2 \to E_1 \uplus E_2$ for the two corresponding inclusions. The disjoint unions of the two FD processes is the process $(E_1 \uplus E_2, \mathcal{E}, (\hat{P}_t), (P_t), \Omega, \mathcal{G}, (\mathbb{P}^x), obs)$ where:
\begin{itemize}
\item the topology on $E_1 \uplus E_2$ is generated by the topologies on $E_1$ and $E_2$: an open set of $E_1 \uplus E_2$ is $i_1 (O_1) \cup i_2 (O_2)$ where $O_1$ and $O_2$ are opens of $E_1$ and $E_2$ respectively,
\item $\mathcal{E}$ is the Borel-algebra generated by this topology. It can also be expressed as the $\sigma$-algebra generated by $\{ i_1(C) ~|~ C \in \mathcal{E}_1 \}$ and $\{ i_2(C) ~|~ C \in \mathcal{E}_2 \}$,
\item for any state $x \in E_1 \uplus E_2$, any time $t \geq 0$ and any function in $C_0(E_1 \uplus E_2)$, we define the semigroup:
\[ \hat{P}_t f(x) = \begin{cases}
\hat{P}_t^1 f_1 (x_1) ~~\text{ if } x = i_1 (x_1)\\
\hat{P}_t^2 f_2 (x_2) ~~\text{ if } x = i_2 (x_2)
\end{cases} \]
where $f_j : E_j \to \mathbb{R}$ is defined by $f_j (y) = f \circ i_j (y)$. The semigroup $\hat{P}_t$ inherits the desired properties from $\hat{P}_t^1$ and $\hat{P}_t^2$ for it to be a FDS,
\item for any state $x \in E_1 \uplus E_2$, any time $t \geq 0$ and any measurable set $C \in \mathcal{E}$, the kernel can be made explicit as:
\[ P_t(x, C) = \begin{cases}
P_t^1 (x_1, i_1^{-1} (C)) ~~\text{ if } x = i_1 (x_1)\\
P_t^2 (x_2, i_2^{-1} (C)) ~~\text{ if } x = i_2 (x_2)
\end{cases},\]
\item for any state $x \in E_1 \uplus E_2$, we set
\[ obs(x) = \begin{cases}
obs_1 (x_1) ~~\text{ if } x = i_1 (x_1)\\
obs_2 (x_2) ~~\text{ if } x = i_2 (x_2)
\end{cases}, \]
\item the set of trajectories on $(E_1 \uplus E_2)_\partial$ is denoted $\Omega$. Note that a trajectory in $\Omega$ can switch between $E_1$ and $E_2$. The set $\Omega$ is equipped with a $\sigma$-algebra $\mathcal{G}$ as is standard for FD processes. For a state $x$, we can explicit the probability distribution for $B \in \mathcal{G}$:
\[ \mathbb{P}^x(B) = \begin{cases}
\mathbb{P}^{x_1}_1 (B_1) ~~\text{ if } x = i_1 (x_1)\\
\mathbb{P}^{x_2}_2 (B_2) ~~\text{ if } x = i_2 (x_2)
\end{cases} \]
where $B_j = \{ \omega \in \Omega_j ~|~ i_j \circ \omega \in B \}$. Note that for any $x \in E_1 \uplus E_2$, for any measurable set $B \subset \{ \omega \in \Omega ~|~ \exists t_1, t_2 ~ \omega (t_1) \in i_1 (E_1) \text{ and } \omega (t_2) \in i_2 (E_2) \}$, $\mathbb{P}^x (B) = 0$.
\end{itemize}

We can also make explicit what a bisimulation is in that context (we will omit to mention the inclusions $i_1$ and $i_2$ to be readable):
\begin{definition}
Given two FD processes
$(E_j, \mathcal{E}_j, (\hat{P}_t^j), (P_t^j), \Omega_j, \mathcal{G}_j, (\mathbb{P}_j^x), obs_j)_{j = 1,2}$, a bisimulation between the two FDPs is an equivalence $R$ on $E_1 \uplus E_2$ such that for all $x R y$ ($x \in E_i$, $y\in E_j$),
\begin{description}
\item[\namedlabel{init1d}{(inititiation 1)}] $obs_i (x) = obs_j(y)$, and
\item[\namedlabel{ind2d}{(induction 2)}]  for all measurable $R$-closed sets $B$, $\mathbb{P}^x(B \cap \Omega_i) =
  \mathbb{P}^y (B \cap \Omega_j)$.\\
  This condition can also be stated as follows. For all sets $B_1 \in \mathcal{G}_1$ and $B_2 \in \mathcal{G}$, $\mathbb{P}_{i}^x (B_i) = \mathbb{P}_{j}^y (B_j)$ if the two sets satisfy the following condition:
  \[ \forall \omega_k \in B_k ~ \forall \omega_l \in \Omega_l ~ (\forall t \geq 0 ~ \omega_k (t) ~ R~ \omega_l(t)) \Rightarrow \omega_l \in B_l \]
  In that formulation, $B_k = B \cap \Omega_k$ and the condition states that the set $B$ is $R$-closed in terms of the sets $B_1$ and $B_2$.
\end{description}
\end{definition}

Note that $R \cap (E_j \times E_j)$ is a bisimulation on $(E_j, \mathcal{E}_j, (P_{t}^j), (\mathbb{P}_{j}^x))$.
To proceed with our cospan idea we need a functional version of bisimulation; we
call these Feller-Dynking homomorphisms or FD-homomorphisms for short.
\begin{definition}
A continuous function $f : E \to E'$ is called a \emph{FD-homomorphism} if it satisfies the following conditions:
\begin{itemize}
\item $obs = obs' \circ f$,
\item for all $x \in E$ and for all measurable sets $B' \subset \Omega '$,
  $\mathbb{P}^{f(x)}(B') = \mathbb{P}^x (B)$ where $B = \{ \omega \in \Omega ~|~
  f \circ \omega \in B'\}$. 
\end{itemize}
\end{definition}

Note that if $f$ and $g$ are FD-homomorphisms, then so is $g \circ f$.

\begin{prop}
The equivalence relation $R$ defined on $E \uplus E'$ as
\[ R = \{ (x,y) \in E \times E ~|~ f(x) = f(y)\} \cup \{ (x,y), (y,x) ~|~ f(x) = y\} \]
is a bisimulation on $E$.
\end{prop}

\begin{proof}
Consider $x$ and $y$ such that $x~R~y$.  We are going to assume that $f(x) =
f(y)$ and we will be treating the case $x R f(x)$ at the same time. 

First note that $obs (x) = obs' \circ f(x)$ since $obs = obs' \circ f$. Since $f(x) = f(y)$, we have that $obs(x) = obs' \circ f(x) = obs' \circ f(y) = obs(y)$. This gives us \ref{init1d} for both cases.

Second, let us check the induction condition \ref{ind2d}. Consider an $R$-closed set $\hat{B}$. Define $B' = \hat{B}\cap \Omega '$. Define $B = \{ \omega \in \Omega ~|~ f \circ \omega \in \hat{B} \cap \Omega'\}$. As $f$ is an FD-homomorphism, we have that $\mathbb{P}^{f(x)} (B') = \mathbb{P}^x (B)$.

Let us show that $B = \hat{B} \cap \Omega$.
\begin{itemize}
\item Consider $\omega \in B$, i.e. $f \circ \omega \in \hat{B} \cap \Omega '$.  By definition $\omega \in \Omega$. Furthermore, $f \circ \omega \in \hat{B}$. By definition of $R$, we have that for all $t \geq 0$, $\omega(t) ~R~ f \circ \omega (t)$. Since the set $\hat{B}$ is $R$-closed and $f \circ \omega \in \hat{B}$, we have that $\omega \in \hat{B}$ which proves the first inclusion.
\item Consider $\omega \in \hat{B} \cap \Omega$. The trajectory $f \circ \omega$ is well-defined and is in $\Omega'$ since $f$ is continuous. Similarly to what was done for the first inclusion, we get that $f \circ \omega \in \hat{B} $ since $\omega \in \hat{B}$ and $\hat{B}$ is $R$-closed. This proves that $\omega \in B$
\end{itemize}
We get that $\mathbb{P}^{f(x)} (\hat{B} \cap \Omega ') = \mathbb{P}^x (\hat{B} \cap \Omega)$.

Since $f(x) = f(y)$, we also get that $\mathbb{P}^x (\hat{B} \cap \Omega) = \mathbb{P}^y (\hat{B} \cap \Omega)$.
\end{proof}

\begin{cor}
The equivalence relation $R$ defined on $E$ as
\[ R = \{ (x,y) \in E \times E ~|~ f(x) = f(y)\} \]
is a bisimulation on $E$.
\end{cor}

Here is an example with one atomic proposition.  Let $\mathcal{M}_1$ be the
standard Brownian motion on the real line with $obs_1 (x) = 1$ if and only if
$x \in \mathbb{Z}$.  Let $\mathcal{M}_2$ be the reflected Brownian motion on
$[0,1]$ with $obs_2 (x) = 1$ if and only if $x = 0$ or $1$.  Let
$\mathcal{M}_3$ be the reflected Brownian motion on $\left[0,\frac{1}{2}\right]$
with $obs_3 (x) = 1$ if and only if $x = 0$.  Let $\mathcal{M}_4$ be the
standard Brownian motion on the circle of radius $\frac{1}{2\pi}$ (we will
identify points on the circle with the angle wrt the vertical) with
$obs_4 (x) = 1$ if and only if $x = 0$.

We can define some natural mappings between these processes:
\[ \xymatrix{
& \mathcal{M}_1 \ar[dl]_{\phi_4} \ar[dr]^{\phi_3} &\\
\mathcal{M}_2 \ar[dr]_{\phi_2} && \mathcal{M}_4 \ar[dl]^{\phi_1}\\
& \mathcal{M}_3& \\
} \]
where
\begin{align*}
\phi_1 : [-\pi, \pi] & \to \left[ 0, 1/2 \right]\\
\theta & \mapsto  | \theta |/2 \pi \\
\phi_2 : [0, 1] & \to \left[ 0, 1/2 \right]\\
x & \mapsto x &\text{if } x \leq 1/2 \\
x & \mapsto 1-x & \text{otherwise}\\
\phi_3 : \mathbb{R} & \to \left[ -\pi , \pi \right]\\
x & \mapsto 2\pi |x - y| &\text{where } y \in \mathbb{Z} \text{ such that } |x -y| \leq 1/2\\
\phi_4 : \mathbb{R} & \to \left[ 0, 1 \right]\\
x & \mapsto x - 2n & \text{if } \exists n \in \mathbb{Z} ~~ 2n \leq x < 2n+1\\
x & \mapsto 2n+2 - x & \text{if } \exists n \in \mathbb{Z} ~~ 2n+1 \leq x < 2n+2
\end{align*}
Note that the condition in the definition of $\phi_3$ means that $y$ is the closest integer to $x$.

\begin{prop}
These morphisms are FD-homomorphisms.
\end{prop}

\begin{proof}
Note that all these functions are continuous.

First note that $obs_3 \circ \phi_1 (\theta) = 1$ if and only if
$\phi_1 (\theta) = 0$ (by definition of $obs_3$). By definition of $\phi_1$,
this is $| \theta | / 2 \pi = 0$, i.e. $\theta = 0$.  But this corresponds to the
only case where $obs_4 (\theta) = 1$ (note that what we have proven is an
equivalence). We have therefore proven that $obs_3 \circ \phi_1 = obs_4$.

Second, $obs_3 \circ \phi_2 (x) = 1$ if and only if $\phi_2 (x) = 0$ (by
definition of $obs_3$) if and only if $x = 0$ or $1$ (by definition of
$\phi_2$) if and only if $obs_2 (x) = 1$ (by definition of $obs_2$).

Third, $obs_2 \circ \phi_4 (x) = 1$ if and only if $\phi_4(x) = 0$ or $1$ (by
definition of $obs_2$). Note that $\phi_4 (x) = 0$ if and only if $x = 2n$ for
some $n \in \mathbb{Z}$. Similarly, $\phi_4 (x) = 1$ if and only if $x = 2n+1$
for some $n \in \mathbb{Z}$. This means that $obs_2 \circ \phi_4 (x) = 1$ if,
and only if, $x \in \mathbb{Z}$ if and only if $obs_1(x) = 1$.

Fourth and finally, $obs_4 \circ \phi_3 (x) = 1$ if and only if
$\phi_3 (x) = 0$ (by definition of $obs_4$) if, and only if $|x - y| = 0$ where
$y \in \mathbb{Z}$ such that $|x - y| \leq 1/2$ if and only if
$x \in \mathbb{Z}$ if and only if $obs_1(x) = 1$.

The second condition is obvious by definition of Brownian motion on these sets.
\end{proof}

\subsection{Definition}

\begin{definition}
A \emph{FD-cospan} is a cospan of FD-homomorphisms.
\end{definition}

\begin{thm}
\label{thm:pushout_FDP}
The category with Feller-Dynkin processes as objects and FD-homomorphisms as morphisms has pushouts.
%If there is a span of dynamic mappings $E_1 \overset{h}{\leftarrow} E_2 \overset{g}{\rightarrow} E_3$, then there is also a cospan of dynamic mappings $E_1 \overset{\phi_1}{\rightarrow} E_2 \overset{\phi_3}{\leftarrow} E_3$.
\end{thm}

\begin{proof}
There are two inclusions $i_1 : E_1 \to E_1 \uplus E_3$ and $i_3 : E_3 \to E_1
\uplus E_3$. Define the equivalence relation $\sim$ on $E_1 \uplus E_3$ as the
smallest equivalence such that for all $z \in E_2$, $i_1 \circ h (z) = i_3 \circ
g(z)$. Define $E_4 = E_1 \uplus E_3 / \sim$ with its corresponding quotient
$\pi_\sim : E_1 \uplus E_3 \to E_4$ and the two maps $\phi_3 = \pi_\sim \circ
i_3$ and $\phi_1 = \pi_\sim \circ i_1$. Note that this corresponds to the
pushout in $Set$. 
\[ \xymatrix{
E_2 \ar[r]^g \ar[d]_h & E_3 \ar[d]^{i_3} \ar@/^1em/[ddr]^{\phi_3}\\
E_1 \ar[r]_{i_1} \ar@/_1em/[drr]_{\phi_1} & E_1 \uplus E_3 \ar[dr]^{\pi_\sim}\\
&& E_4
} \]

We equip this set with the smallest topology that makes $\pi_\sim$ continuous
(where the topology on $E_1 \uplus E_3$ is the topology inherited from the
inclusions). Note that this corresponds to the pushout in $Top$. 

We define $obs_4$ as such:
\[ obs_4 (x_4) = \begin{cases}
obs_1 (x_1) \qquad \text{if } x_4 = \phi_1 (x_1)\\
obs_3 (x_3) \qquad \text{if } x_4 = \phi_3 (x_3)
\end{cases} \]
In order to prove that this is well-defined, we have to prove that if $\phi_1
(x_1) = \phi_3 (x_3)$, then $obs_1 (x_1) = obs_3 (x_3)$. Assume $\phi_1 (x_1) =
\phi_3 (x_3)$, this means that there exists $x_2 \in E_2$ such that $h(x_2) =
x_1$ and $g(x_2) = x_3$, and as $g$ and $h$ are FD-homomorphisms, we have that 
\[ obs_1 (x_1) = obs_1 h (x_2) = obs_2(x_2) = obs_3 g(x_2) = obs_3 (x_3) \]

Finally, let us define the Feller-Dynkin Process. Let $\Omega_4$ be the set of
R-paths on $E_4$. It is equipped with a $\sigma$-algebra $\mathcal{G}_4$ defined
in the standard way for Feller-Dynkin processes. Consider $B_4 \in
\mathcal{G}_4$. We define 
\[ \mathbb{P}^x_4 (B_4) = \begin{cases}
\mathbb{P}^z_1 (B_1) \qquad \text{if } x = \phi_1 (z)\\
\mathbb{P}^z_3 (B_3) \qquad \text{if } x = \phi_3 (z)
\end{cases} \]

where $B_1 = \{ \omega \in \Omega_1 ~|~ \phi_1 \circ \omega \in B_4\}$ and $B_3
= \{ \omega \in \Omega_3 ~|~ \phi_3 \circ \omega \in B_4\}$. In order to prove
that this is well-defined, we have to prove first that $B_1 \in \mathcal{G}_1$
and $B_3 \in \mathcal{G}_3$ and second that if $\phi_1(z_1) = \phi_3(z_3) = x$,
then $\mathbb{P}^{z_1}_1 (B_1) = \mathbb{P}^{z_3}_3 (B_3)$. 

Let us start with the measurability of $B_1$ ($B_3$ works exactly in the same
way). Let us write $X^1$ and $X^4$ for the random variables associated to the
FDP on $E_1$ and $E_4$. Recall that  
\[\mathcal{G}_i = \sigma (X^i_s ~|~ 0 \leq s < \infty) = \sigma(\{ (X^i_s) ^{-1}(C) ~|~ 0 \leq s < \infty, ~ C \in \mathcal{E}_i \})\]
We are going to prove by induction on the structure of $B_4$ that $B_1 \in \mathcal{G}_1$.
\begin{itemize}
\item First, if $B_4 = (X_s^4)^{-1} (C) = \{ \omega \in \Omega_4 ~|~ \omega(s)
  \in C\}$ for $C \in \mathcal{E}_4$, then 
\begin{align*}
B_1
& = \{ \omega \in \Omega_1 ~|~ \phi_1 \circ \omega \in B_4\}\\
& = \{ \omega \in \Omega_1 ~|~ \phi_1 \circ \omega (s) \in C\}\\
& = \{ \omega \in \Omega_1 ~|~ \omega (s) \in \phi_1^{-1}(C)\}\\
& = (X_s^1)^{-1} (\phi_1^{-1}(C))
\end{align*}
And since $\phi_1$ is continuous, we know that
$\phi_1^{-1}(C) \in \mathcal{E}_1$ and hence $B_1 \in \mathcal{G}_1$.
\item If $B_4 = A_4^C$ with $A_4 \in \mathcal{G}_4$ and $A_1 = \{ \omega \in \Omega_1 ~|~ \phi_1 \circ \omega \in A_4\} \in \mathcal{G}_1$, then
\begin{align*}
B_1
& = \{ \omega \in \Omega_1 ~|~ \phi_1 \circ \omega \in B_4\}\\
& = \{ \omega \in \Omega_1 ~|~ \phi_1 \circ \omega \notin A_4\}\\
& = \Omega_1 \setminus \{ \omega \in \Omega_1 ~|~ \phi_1 \circ \omega \in A_4\}\\
& = \Omega_1 \setminus A_1
\end{align*}
And since $A_1 \in \mathcal{G}_1$, we get that $B_1 \in \mathcal{G}_1$.
\item If $B_4 = A_1 \cup A_2 \cup ...$ where for every $i \in \mathbb{N}$, $A_i \in \mathcal{G}_4$ and $A'_i = \{ \omega \in \Omega_1 ~|~ \phi_1 \circ \omega \in A_i \} \in \mathcal{G_1}$, then
\begin{align*}
B_1
& = \{ \omega \in \Omega_1 ~|~ \phi_1 \circ \omega \in B_4\}\\
& = \{ \omega \in \Omega_1 ~|~ \phi_1 \circ \omega \in A_1 \cup A_2 \cup ...\}\\
& = \bigcup_{n \in \mathbb{N}} \{ \omega \in \Omega_1 ~|~ \phi_1 \circ \omega \in A_i\}\\
& = \bigcup_{n \in \mathbb{N}} A'_i
\end{align*}
And since $A'_i \in \mathcal{G}_1$ for every $i \in \mathbb{N}$, we get that $B_1 \in \mathcal{G}_1$.
\end{itemize}
This proves that $B_1 \in \mathcal{G}_1$.

First note that if $\phi_1(z_1) = \phi_3(z_3)$, then there exists $z_2 \in E_2$ such that $z_1 = h(z_2)$ and $z_3 = g(z_2)$. Since $h$ and $g$ are FD-homomorphisms, we get that
\begin{align*}
\mathbb{P}^{z_1}_1 (B_1)
& = \mathbb{P}^{z_2}_2 \left( \{ \omega \in \Omega_2 ~|~ h \circ \omega \in B_1 \}\right) \text{ as $h$ is a FD-homomorphism}\\
& = \mathbb{P}^{z_2}_2 \left( \{ \omega \in \Omega_2 ~|~ \phi_1 \circ h \circ \omega \in B_4 \}\right) \text{ by definition of $B_1$}\\
& = \mathbb{P}^{z_2}_2 \left( \{ \omega \in \Omega_2 ~|~ \phi_3 \circ g \circ \omega \in B_4 \}\right) \text{ as }\phi_1 \circ h = \phi_3 \circ g\\
& = \mathbb{P}^{z_2}_2 \left( \{ \omega \in \Omega_2 ~|~ g \circ \omega \in B_3 \}\right) \text{ by definition of $B_3$}\\
& = \mathbb{P}^{z_3}_3 (B_3) \text{ as $g$ is a FD-homomorphism}
\end{align*}
This indeed defines a probability distribution on $(\Omega_4, \mathcal{G}_4)$ (this is a direct consequence of the fact that $\mathbb{P}^{z_1}$ and $\mathbb{P}^{z_3}$ are probability distributions).

Let us now check that this is indeed a Feller-Dynkin Process. We define first the corresponding kernel for $t \geq 0$, $x \in E_4$ and $C \in \mathcal{E}_4$,
\[ P_t^4 (x, C) = \mathbb{P}^x_4 (\{ \omega ~|~ \omega (t) \in C\}) \]
Note that if $x = \phi_1(z)$, $P_t^4 (x, C) =P_t^1 (z, \phi_1^{-1}(C))$ (and with the corresponding equality if $x = \phi_3(z)$).
We can now define the operator
\[ \hat{P}_t^4 f (x) = \int_{y \in E_4} f(y) P_t^4 (x, dy) \]
Using the corresponding equality for the kernel and a change of measure, we get that if $x = \phi_1(z)$, then $\hat{P}_t^4 f (x) = \hat{P}_t^1 (f \circ \phi_1)(z)$ (and with the corresponding equality if $x = \phi_3(z)$). Now it is obvious that the process defined is indeed a Feller-Dynkin Process since both $\hat{P}^1$ and $\hat{P}^3$ are.

By construction, $\phi_1$ and $\phi_3$ are indeed FD-homomorphisms.

Let us now check the universal property. Assume there is a Feller-Dynkin process $(E_5, \mathcal{E}_5, (P_t^5)_t)$ and two FD-homomorphisms $\psi_1 : E_1 \to E_5$ and $\psi_3 : E_3 \to E_5$ such that $\psi_1 \circ h = \psi_3 \circ g$. Since $E_4$ is defined as the corresponding pushout in $Set$, there exists a morphism $\gamma : E_4 \to E_5$ in $Set$ such that $\psi_j = \gamma \circ \phi_j$ (for $j = 1,2$). Let us show that it is a FD-homomorphism.

Since $(E_4, \phi_1, \phi_3)$ is the pushout in $Top$, we know that $\gamma$ is continuous.

First, we want to show that $obs_4 = obs_5 \circ \gamma$. Let $x_4 \in E_4$, there are two cases: $x_4 = \phi_1 (x_1)$ or $x_4 = \phi_3 (x_3)$. Consider the first case $x_4 = \phi_1 (x_1)$ (the other case is similar).
\begin{align*}
obs_5 \circ \gamma (x_4)
& = obs_5 \circ \gamma \circ \phi_1 (x_1)\\
& = obs_5 \circ \psi_1 (x_1)\\
& = obs_1 (x_1) \qquad \text{ since $\psi_1$ is a FD-homomorphism}\\
& = obs_4 (\phi_1 (x_1)) \qquad \text{ since $\phi_1$ is a FD-homomorphism}\\
& = obs_4 (x_4)
\end{align*}

Let us now show that for $x_4 \in E_4$ and $B_5 \in \mathcal{G}_5$, $\mathbb{P}_5^{\gamma (x_4)} (B_5) = \mathbb{P}_4^{x_4} (B_4)$ with $B_4 = \{ \omega \in \Omega_4 ~|~ \gamma \circ \omega \in B_5 \}$. Consider the first case $x_4 = \phi_1 (x_1)$ (the other case $x_4 = \phi_3 (x_3)$ is similar). First, as $\psi_1$ is a FD-homomorphism,
\begin{align*}
\mathbb{P}_5^{\gamma (x_4)} (B_5)
& = \mathbb{P}_5^{\gamma (\phi_1 (x_1))} (B_5) = \mathbb{P}_5^{\psi_1 (x_1)} (B_5)\\
& = \mathbb{P}_1^{x_1} (B_1) \quad \text{ where } B_1 = \{ \omega \in \Omega_1 ~|~ \psi_1 \circ \omega \in B_5\}
\end{align*}
Second, as $\phi_1$ is a FD-homomorphism,
\[ \mathbb{P}_4^{x_4} (B_4) = \mathbb{P}_4^{\phi_1(x_1)} (B_4) = P_1^{x_1}(B'_1) \]
where $B'_1 = \{ \omega \in \Omega_1 ~|~ \phi_1 \circ \omega \in B_4\}$.
Finally, we have that
\begin{align*}
B'_1 & = \{ \omega \in \Omega_1 ~|~ \phi_1 \circ \omega \in B_4\}\\
& = \{ \omega \in \Omega_1 ~|~ \gamma \circ \phi_1 \circ \omega \in B_5\} = B_1
\end{align*}
which concludes the proof.
\end{proof}

This proves that FD-cospan corresponds to an equivalence relation. We have
already showed how FD-homomorphisms (and hence FD-cospans) yield
bisimulations. Let us now show that the converse is also true. 

\begin{thm}
For all bisimulations $R$, there exists $(f,g)$ a FD-cospan such that
\begin{itemize}
\item for all $x \in E_1, y \in E_2$, $x~R~y$ if and only if $f(x) = g(y)$,
\item for all $x,x' \in E_1$, $x ~R~ x'$ if and only if $f(x) = f(x')$, and
\item for all $y,y' \in E_2$, $y ~R~ y'$ if and only if $g(y) = g(y')$.
\end{itemize}
\end{thm}

\begin{proof}
Define the set $E = (E_1 \uplus E_2) / R$. There are two inclusions $i_1 : E_1 \to E_1 \uplus E_2$, $i_2 : E_2 \to E_1 \uplus E_2$ and a quotient $\pi_R : E_1 \uplus E_2 \to E$. We define $f = \pi_R \circ i_1 : E_1 \to E$ and $g = \pi_R \circ i_2 : E_2 \to E$.
We can equip the set $E$ with the smallest topology that makes all those maps continuous.

Let us now clarify what the FDP is on that state space. First, we define for $x \in E$
\[ obs(x) = \begin{cases}
obs_1 (x_1) \qquad \text{if } x = f(x_1)\\
obs_2(x_2) \qquad \text{if } x = g(x_2)
\end{cases} \]
This is indeed well-defined as whenever $y R z$ for $y \in E_i$ and $z \in E_j$, then $obs_i(y) = obs_j(z)$.

Let us denote $\Omega$ the set of trajectories on $E_\partial$. Its $\sigma$-algebra $\mathcal{G}$ is defined as usual for FD processes. 
We define the following family of
probabilities on the set $\Omega$: for $x \in E$, $B \in \mathcal{G}$, 
\[ \mathbb{P}^x(B) = \begin{cases}
\mathbb{P}_1^{x_1}(B_1) \qquad \text{if } x = f(x_1) \text{ with } B_1 = \{
\omega \in \Omega_1 ~|~ f \circ \omega \in B\}\\ 
\mathbb{P}_2^{x_2}(B_2) \qquad \text{if } x = g(x_2) \text{ with } B_2 = \{
\omega \in \Omega_2 ~|~ g \circ \omega \in B\} 
\end{cases} \]
This is well-defined. Indeed, let us show that the sets $B_1$ and $B_2$ satisfy the condition of \ref{ind2d}.
\begin{itemize}
\item Take $\omega_1 \in B_1$, i.e. $f \circ \omega_1 \in B$ and $\omega_2 \in \Omega_2$. If for all time $t \geq 0$, $\omega_1 (t) ~R~ \omega_2(t)$, this means that $\forall t \geq 0 ~ f \circ \omega_1(t) = g \circ \omega_2(t)$, i.e. $f \circ \omega_1 = g \circ \omega_2$. Since $f \circ \omega_1 \in B$, we get that $g \circ \omega_2 \in B$ and therefore $\omega_2 \in B_2$.
\item Take $\omega_1 \in B_1$ and $\omega'_1 \in \Omega_1$. If $\forall t \geq 0 ~ (\omega_1(t) ~R~ \omega'_1(t))$, then $\forall t \geq 0 ~ f \circ \omega_1(t) = f \circ \omega'_1(t)$, i.e. $f \circ \omega_1 = f \circ \omega'_1$ and therefore $\omega'_1 \in B_1$.
\item Similarly we get the other conditions.
\end{itemize}

\begin{itemize}
\item Assume $f(x_1) = f(y_1)$. This means that $x_1 ~R~ y_1$ and applying condition \ref{ind2d}, $\mathbb{P}_1^{x_1} (B_1) = \mathbb{P}_1^{y_1} (B_1)$.
\item Assume $f(x_1) = g(x_2)$. This means that $x_1 ~R~ x_2$ and applying condition \ref{ind2d}, $\mathbb{P}_1^{x_1} (B_1) = \mathbb{P}_2^{x_2} (B_2)$.
\item Similarly we get the other conditions.
\end{itemize}

The Markov kernel of the FD process is defined as such: for $x \in E$, $C \in \mathcal{E}$ and $t \geq 0$
\[ P_t(x, C) = \begin{cases}
P^1_t(x_1, f^{-1}(C)) \qquad \text{if } x = f(x_1)\\
P^2_t(x_2, g^{-1}(C)) \qquad \text{if } x = g(x_2)
\end{cases} \]

By construction,
\begin{itemize}
\item for all $x \in E_1, y \in E_2$, $x~R~y$ if and only if $f(x) = g(y)$,
\item for all $x,x' \in E_1$, $x ~R~ x'$ if and only if $f(x) = f(x')$, and
\item for all $y,y' \in E_2$, $y ~R~ y'$ if and only if $g(y) = g(y')$.
\end{itemize}
and it is easy to see that we have proven that $f$ and $g$ are FD-homomorphisms when checking that $\mathbb{P}^x$ is well-defined.
\end{proof}

We thus have the following correspondance between bisimulation and FD-cospans:

\begin{theorem}
Two states $x$ and $y$ are bisimilar if and only if there exists a FD-cospan $(f,g)$ such that $f(x) = g(y)$.
\end{theorem}

\section{Comparison to discrete time bisimulation}
\label{sec:DT}

The goal of this work is to extend the notion of bisimulation that exists in discrete time to a continuous-time setting. Therefore an important question is the following: do we get back the definition of bisimulation that existed in discrete time when we restrict Feller-Dynkin processes to (some kind of) discrete-time processes?

Given an LMP $(X, \Sigma, \tau, (\chi_A)_{A \in AP})$, we can always view it as a FD process on $(E, \mathcal{E})$ with
$ E = X \times [0, 1)$ and $\mathcal{E} = \Sigma \times \mathcal{B}([0, 1))$ by
adding to the space the following kernel: for all $x \in X$ and $C \in \Sigma$,
$t \geq 0$ and $s \in [0,1)$,
$P_t ((x,s), C) = \tau_{\lfloor t + s\rfloor} (x, C')$ where
$C' = \{ z ~|~ (z,t+s - \lfloor t+s \rfloor) \in C\}$ and for $k \geq 1$,
\begin{align*}
\tau_0 (x, C) & = \delta _C(x)\\
\tau_1(x,C) & = \tau (x,C)\\
\tau_{k + 1}(x,C) & = \int _{y \in X} \tau(x,dy) \tau_k(y ,C)
\end{align*}
We also define $(obs(x,s))_i = \chi_{A_i}(x)$ (where $AP = \{ A_1, A_2, ...\}$).

%\begin{claim}
%Given an LMP viewed as a FD process, for all $x \in X$, $0 \leq s \leq t < 1$ and $B \in \mathcal{G}$,
%\[ \mathbb{P}^{(x,s)}(B) = \mathbb{P}^{(x,t)}(\partial B) \]
%where $\partial B = \{ \omega' ~|~ \exists \omega \in B' ~ \forall u \geq 0 ~ \omega ' (u) = \omega (u + t-s)\}$
%\end{claim}

Let us recall the definition of bisimulation in the discrete time setting.

\begin{definition}
Given an LMP $(X, \Sigma, \tau, (\chi_A)_{A \in AP})$, a \emph{DT-bisimulation}
$R$ is an equivalence relation on $X$ such that if $x Ry$, then 
\begin{itemize}
\item for all $A \in AP$, $\chi_A(x) = \chi_A(y)$
\item for all $R$-closed set $B \in \Sigma$, $\tau(x, B) = \tau(y, B)$.
\end{itemize}
\end{definition}

\begin{lemma}
\label{lemma:DTbisim_nsteps}
Consider a DT-bisimulation $R$. If $x R y$, then for all $n \geq 1$, for all
$R$-closed set $A_1, ..., A_n$,
\[ \int_{x_1 \in A_1} ... \int_{x_n \in A_n} \tau (x, dx_1) \tau (x_1, dx_2)... \tau(x_{n-1}, dx_n) =
\int_{x_1 \in A_1} ... \int_{x_n \in A_n} \tau (y, dx_1) \tau (x_1, dx_2)... \tau(x_{n-1}, dx_n) \]
\end{lemma}

\begin{proof}

Let us denote $\pi_R : X \to X/R$ the quotient. We can also define some function
\[ \overline{\tau} (\pi_R(x), A) = \tau (x, \pi_R^{-1}(A)) \]
Note that the choice of $x$ does not change the right term since $R$ is a DT-bisimulation and $\pi_R^{-1}(A))$ is an $R$-closed set. A sequence of change of variables yields:
\begin{align*}
& \int_{x_1 \in A_1} ... \int_{x_n \in A_n} \tau (x, dx_1) \tau (x_1, dx_2)... \tau(x_{n-1}, dx_n)\\
& = \int_{x_1 \in A_1} ... \int_{x_{n-1} \in A_{n-1}} \int_{y_n \in A_n/R} \tau (x, dx_1) \tau (x_1, dx_2)... \tau(x_{n-2}, dx_{n-1}) \tau(x_{n-1}, \pi_R^{-1}(dy_n))\\
& = \int_{x_1 \in A_1} ... \int_{x_{n-1} \in A_{n-1}} \int_{y_n \in A_n/R} \tau (x, dx_1) \tau (x_1, dx_2)... \tau(x_{n-2}, dx_{n-1}) \overline{\tau}(\pi_R(x_{n-1}), dy_n)\\
& = \int_{x_1 \in A_1} ... \int_{x_{n-2} \in A_{n-2}} \int_{y_{n-1} \in A_{n-1}/R} \int_{y_n \in A_n/R} \tau (x, dx_1) \tau (x_1, dx_2)... \tau(x_{n-2}, \pi_R^{-1}(dy_{n-1})) \overline{\tau}(y_{n-1}, dy_n)\\
& = \int_{x_1 \in A_1} ... \int_{x_{n-2} \in A_{n-2}} \int_{y_{n-1} \in A_{n-1}/R} \int_{y_n \in A_n/R} \tau (x, dx_1) \tau (x_1, dx_2)... \overline{\tau}(\pi_R(x_{n-2}), dy_{n-1}) \overline{\tau}(y_{n-1}, dy_n)\\
& = \int_{y_1 \in A_1/R} ... \int_{y_n \in A_n/R} \tau (x, \pi_R^{-1}(dy_1)) \overline{\tau} (y_1, dy_2)... \overline{\tau} (y_{n-1}, dy_n)
\end{align*}
And since the two measures $\tau (x, \pi_R^{-1}(\bullet))$ and $\tau (y, \pi_R^{-1}(\bullet))$ are equal as $R$ is a DT-bisimulation, this concludes the proof
%Let us denote $\pi_R : X \to X/R$ the quotient. We can also define some function $r
%: X /R \to X$ such that $\pi_R r([z]) = [z]$. There are several such functions
%that exist, and we will have to address that later on. 
%\begin{align*}
%\tau_{n+1}(x,C)
%& = \int_{z \in X} \tau(x, dz) \tau_n(z, C)\\
%& = \int_{[z] \in X / R} \tau_n(r([z]), C) \tau(x, \pi_R^{-1}(d [z]))
%\end{align*}
%First note that since $x R y$, $\tau(x, \pi_R^{-1}(.)) = \tau(y, \pi_R^{-1}(.))$. Second, let us address the problem raised on $r$: if we have two functions $r_1$ and $r_2$ (instead of $r$), then $r_1([z]) = r_2([z])$ which means in particular that $\tau_n(r_1([z]), C) = \tau_n(r_2([z]), C)$, so $r$ can be defined however.
%This proves that $\tau_{n+1}(x,C) = \tau_{n+1}(y,C)$.
\end{proof}

\begin{proposition}
If the equivalence $R$ is a DT-bisimulation, then the relation $R'$ defined as
\[ R' = \{ \left( (x,s), (y,s) \right) ~|~ s \in [0,1), x~R~y \} \]
is a bisimulation.
\end{proposition}

\begin{proof}
First note that the relation $R'$ is indeed an equivalence since $R$ is one.

Assume $(x,s) R' (y,s)$.

Let us show that $obs(x,s) = obs(y,s)$. This is a direct consequence of the fact that $obs_i (x,s) = \chi_{A_i}(x)$ (and similarly for $y$) and since $x~R~y$, $\chi_{A_i}(x) = \chi_{A_i}(y)$.

%Let us consider an $R'$-closed set $C$ and $t \geq 0$.
%\begin{align*}
%P_t((x,s), C)
% & = \tau_{\lfloor t+s \rfloor} (x, C') \qquad \text{where } C' = \{ z ~|~ (z, t+s - \lfloor t+s \rfloor) \in C\}\\
% & = \tau_{\lfloor t+s \rfloor} (y, C')\qquad \text{ using claim \ref{claim:DTbisim_nsteps} and $C'$ $R$-closed}\\
% & = P_t((y,s), C)
%\end{align*}

Finally, we want to prove that $\mathbb{P}^{(x,s)} (B) = \mathbb{P}^{(y,s)} (B)$ for any measurable $R'$-closed set $B$.

The measurable $R'$-closed set $B$ is of the form $\{ \omega ~|~ \forall i \in \mathbb{N}~ \omega (i) \in A_i \}$ where $A_i \in \mathcal{E}_\partial$ is an $R'$-closed set.

Let us denote for all $n \in \mathbb{N}$, $B_n = \{ \omega ~|~ \forall i \leq n ~ \omega (i) \in A_i \}$ and let us show that $\mathbb{P}^{(x,s)} (B_n) = \mathbb{P}^{(y,s)} (B_n)$. To prove that, we can also assume that $\forall (x,t) \in A_i$, $t = s$. We denote $A'_i = \{ z ~|~ (z,s) \in A_i \}$.

Assume $\partial \in A_k$ with $k \leq n$. We can write explicitly:
\begin{align*}
\mathbb{P}^{(x,s)} (B_n)
& = \int_{x_0 \in A_0}... \int_{x_n \in A_n} \delta_{(x,s)}(dx_0) P_1^{+ \partial}(x_0, dx_1) ... P_1^{+ \partial} (x_{n-1}, dx_n)\\
& = \int_{x_0 \in A_0}... \left( \int_{x_k \in A_k \setminus \{ \partial\}}... \int_{x_n \in A_n} \delta_{(x,s)}(dx_0) P_1^{+ \partial}(x_0, dx_1) ... P_1^{+ \partial} (x_{n-1}, dx_n) \right.\\
& \qquad \left. + \int_{x_k \in \{ \partial\}}... \int_{x_n \in A_n} \delta_{(x,s)}(dx_0) P_1^{+ \partial}(x_0, dx_1) ... P_1^{+ \partial} (x_{n-1}, dx_n) \right)\\
& = \mathbb{P}^{(x,s)} (\{ \omega ~|~ \forall j \leq n ~ \omega (j) \in A_j \text{ and } \omega (k) \neq \partial \})\\
& \qquad + \int_{x_0 \in A_0}... \int_{x_k \in \{ \partial\}}... \int_{x_n \in A_n} \delta_{(x,s)}(dx_0) P_1^{+ \partial}(x_0, dx_1) ... P_1^{+ \partial} (x_{n-1}, dx_n) \\
& = \mathbb{P}^{(x,s)} (\{ \omega ~|~ \forall j \leq n ~ \omega (j) \in A_j \text{ and } \omega (k) \neq \partial \})\\
& \qquad + \int_{x_0 \in A_0}... \int_{x_{k-1} \in A_{k-1}} \delta_{(x,s)}(dx_0) P_1^{+ \partial}(x_0, dx_1) ... P_1^{+\partial} (x_{k-1}, \{ \partial\}) \delta_{A_{k+1}} (\partial)... \delta_{A_n} (\partial) \\
& = \mathbb{P}^{(x,s)} (\{ \omega ~|~ \forall j \leq n ~ \omega (j) \in A_j \text{ and } \omega (k) \neq \partial \})\\
& \qquad + \mathbb{P}^{(x,s)} (\{ \omega ~|~ \forall j <k ~ \omega (j) \in A_j \text{ and } \omega (k) = \partial \}) \delta_{A_{k+1}} (\partial)... \delta_{A_n} (\partial)
\end{align*}

For that reason, we can deal with the two following cases and conclude in full generality that for all $n \in \mathbb{N}$, $\mathbb{P}^{(x,s)}(B_n) = \mathbb{P}^{(y,s)}(B_n)$:
\begin{itemize}
\item First, if for all $i \leq n$, $\partial \notin A_i$, then we have that $P_1^{+ \partial} ((z,s), A_i) = \tau (z, A'_i)$ and in this case we have that
\begin{align*}
\mathbb{P}^{(x,s)} (B_n)
& = \int_{x_0 \in A_0}... \int_{x_n \in A_n} \delta_{(x,s)}(dx_0) P_1^{+ \partial}(x_0, dx_1) ... P_1^{+ \partial} (x_{n-1}, dx_n)\\
& = \int_{x'_1 \in A'_1}... \int_{x'_n \in A'_n} \delta_{A'_0}(x) \tau(x, dx'_1) ... \tau (x'_{n-1}, dx'_n)\\
& = \int_{x'_1 \in A'_1}... \int_{x'_n \in A'_n} \delta_{A'_0}(y) \tau(y, dx'_1) ... \tau (x'_{n-1}, dx'_n)\\
& = \mathbb{P}^{(y,s)} (B_n)
\end{align*}
using lemma \ref{lemma:DTbisim_nsteps} and since $A'_0$ is $R$-closed.
\item Second, if $A_n = \{ \partial \}$ and for all $i < n$, $\partial \notin A_i$, then
\begin{align*}
\mathbb{P}^{(x,s)} (B_n)
& = \int_{x_0 \in A_0}... \int_{x_{n-1} \in A_{n-1}} \delta_{(x,s)}(dx_0) P_1^{+ \partial}(x_0, dx_1) ... P_1^{+ \partial} (x_{n-2}, dx_{n-1}) P_1^{+ \partial} (x_{n-1}, \{\partial\})\\
& = \int_{x'_1 \in A'_1}... \int_{x'_{n-1} \in A'_{n-1}} \delta_{A'_0}(x) \tau(x, dx'_1) ... \tau (x'_{n-2}, dx'_{n-1}) (1 - \tau(x'_{n-1}, X))\\
& = \int_{x'_1 \in A'_1}... \int_{x'_{n-1} \in A'_{n-1}} \delta_{A'_0}(x) \tau(x, dx'_1) ... \tau (x'_{n-2}, dx'_{n-1})\\
& \qquad - \int_{x'_1 \in A'_1}... \int_{x'_{n-1} \in A'_{n-1}} \int_{x'_n \in X} \delta_{A'_0}(x) \tau(x, dx'_1) ... \tau (x'_{n-2}, dx'_{n-1}) \tau(x'_{n-1}, dx'_n))\\
& = \int_{x'_1 \in A'_1}... \int_{x'_{n-1} \in A'_{n-1}} \delta_{A'_0}(y) \tau(y, dx'_1) ... \tau (x'_{n-2}, dx'_{n-1})\\
& \qquad - \int_{x'_1 \in A'_1}... \int_{x'_{n-1} \in A'_{n-1}} \int_{x'_n \in X} \delta_{A'_0}(y) \tau(y, dx'_1) ... \tau (x'_{n-2}, dx'_{n-1}) \tau(x'_{n-1}, dx'_n))\\
& = \mathbb{P}^{(y,s)} (B_n)
\end{align*}
using lemma \ref{lemma:DTbisim_nsteps} and since $A'_0$ is $R$-closed.
\end{itemize}

Moreover, since $\mathbb{P}^{(x,s)}(B) = \lim_{n \to \infty} \mathbb{P}^{(x,s)}(B_n)$ (and similarly for $y$), we get the desired result.
\end{proof}

\begin{definition}
An equivalence $R$ on the state space of an LMP viewed as a FD process is \emph{time-coherent} if for all $x,y$ in the state space of the LMP and for all $0 \leq t < 1$,
\[ (x,t) R (y,t) \Rightarrow \forall s \in [ 0,1 ) ~ (x,s) R (y,s) \]
Given any equivalence $R$ on the state space of an LMP viewed as a FD process, we define its \emph{time-coherent closure} (denoted $time(R)$) as the smallest time-coherent equivalence containing $R$.
\end{definition}

\begin{prop}
If $R$ is a bisimulation on an LMP viewed as a FD process, then so is $time(R)$.
\end{prop}

\begin{proof}
Let us first start by clarifying what the equivalence $time(R)$ is.

Define the relation $Q = \{ ((x,s),(y,s)) ~|~ \exists t ~ (x,t) R (y,t) \}$, it is reflexive and symmetric. Let us consider its transitive closure $tc(Q)$. The relation $tc(Q)$ is an equivalence. Moreover, it contains the equivalence $R$ and it is is time-coherent.

Let $R'$ be a time-coherent equivalence containing $R$. We now want to show that $tc(Q) \subset R'$. Let us consider $((x,s),(y,s)) \in tc(Q)$ This means that there exists $n \in \mathbb{N}$, $(x_i)_{i = 0,.., n}$ and $(t_i)_{i = 0, ..., n-1}$ such that $x_0 =x$, $x_n = y$ and $(x_i, t_i) R (x_{i+1}, t_i)$ for all $0 \leq i \leq n-1$. Since $R \subset R'$, for all $i$, $(x_i, t_i) R' (x_{i+1}, t_i)$. Since $R'$ is time-coherent, for all $i$, $(x_i, s) R' (x_{i+1}, s)$. By transitivity of $R'$, $(x_0,s)R' (x_n,s)$, i.e. $(x,s) R'(y,s)$. This proves that $tc(Q)$ is a subset of all time-coherent equivalences containing $R$. Moreover, since $tc(Q)$ is itself a time-coherent equivalence containing $R$, we get that $tc(Q)$ is the smallest such equivalence, i.e. $tc(Q) = time(R)$.

There now remains to prove that $time(R)$ is a bisimulation.
Consider $((x,s),(y,s)) \in time(R) $, this means that there exists $n \in \mathbb{N}$, $(x_i)_{i = 0,.., n}$ and $(t_i)_{i = 0, ..., n-1}$ such that $x_0 =x$, $x_n = y$ and $(x_i, t_i) R (x_{i+1}, t_i)$ for all $0 \leq i \leq n-1$.

First, note that this means that for all $0 \leq i \leq n-1$, $obs(x_i, t_i) = obs(x_{i+1}, t_i)$ since $R$ is a bisimulation. And since $(obs(x_i, t))_j = \chi_{A_j} (x_i)$ for all time $t$ (and similarly for $x_{i+1}$), we get that $obs(x_i, s) = obs(x_{i+1}, s)$ for all $0 \leq i \leq n-1$ and in particular,  $obs(x, s) = obs(y, s)$.

%Let $C$ be a $time(R)$-closed set and let $t \geq 0$. We want to show that $P_t((x,s), C) = P_t((x,s), C)$. First recall that $P_t((x,s), C) = \tau_{\lfloor s + t \rfloor} (x,C)$ (and similarly for $y$). The set $C$ is also $R$-closed as $R \subset time(R)$. This means that for all $s_i \geq 0$, $P_{s_i}((x_i, t_i), C) = P_{s_i}((x_{i+1}, t_i), C)$, i.e. $\tau_{\lfloor s_i + t_i \rfloor} (x_i,C) = \tau_{\lfloor s_i + t_i \rfloor} (x_{i+1},C)$. In particular, we can consider some $s_i$ such that $\lfloor s_i + t_i \rfloor = \lfloor s + t \rfloor$, in which case we get that for all $i$, $\tau_{\lfloor s + t \rfloor} (x_i,C) = \tau_{\lfloor s + t \rfloor} (x_{i+1},C)$. And then $\tau_{\lfloor s + t \rfloor} (x_0,C) = \tau_{\lfloor s + t \rfloor} (x_n,C)$ which proves that $P_t((x,s), C) = P_t((x,s), C)$ (as $x_0 = x$ and $x_n = y$).

Let $B$ be a $time(R)$-closed set.
We want to show that for all $i$, $\mathbb{P}^{(x_i, s)}(B) = \mathbb{P}^{(x_{i+1}, s)}(B)$. As we are working on an LMP viewed as a FD process, a $time(R)$-closed set is of the form
\[ B = \{ \omega ~|~ \forall i \in \mathbb{N} ~ \omega(i) \in A_i \} \]
where $A_i$ is a $time(R)$-closed subset of the state space $A_i$.
In this case, $\mathbb{P}^{(z, s)}(B) = \mathbb{P}^{(z, t_i)}(B)$ for any $z \in X$. Moreover, the sets $A_j$ are also $R$-closed (since $R \subset time(R)$) and hence the set $B$ is $R$-closed. This means that $ \mathbb{P}^{(x_i, t_i)}(B) =  \mathbb{P}^{(x_{i+1}, t_i)}(B)$ which proves our point and therefore that $\mathbb{P}^{(x, s)}(B) = \mathbb{P}^{(y, s)}(B)$
%where $S_0$ and all $Q_m$ are in the complete lattice generated on $X$ by the atomic properties.
%We have that for all $z \in X$, $\mathbb{P}^{(z, s)}(B) = \mathbb{P}^{(z, t_i)}(B_i)$ where
%\[ B_i = \{ \omega ~|~ \omega (0) \in S_0 \times \{ t_i\} \text{ and } \forall m \in \mathbb{N}~ \omega (m + 1-t_i) \in Q_m \times \{ 0 \} \} \]
%The set $B_i$ is $obs$-closed. Since $(x_i, t_i) R (x_{i+1}, t_i)$, we have that $\mathbb{P}^{(x_i, t_i)} (B_i) = \mathbb{P}^{(x_{i+1}, t_i)} (B_i)$. This means that $\mathbb{P}^{(x_i, s)}(B) = \mathbb{P}^{(x_{i+1}, s)}(B)$ for all $i = 0,..., n$ and hence $\mathbb{P}^{(x_0, s)}(B) = \mathbb{P}^{(x_n, s)}(B)$, i.e. $\mathbb{P}^{(x, s)}(B) = \mathbb{P}^{(y, s)}(B)$.
\end{proof}

\begin{thm}
If the equivalence $R$ is a time-coherent bisimulation, then the relation $R'$ defined as
\[ R' = \{ (x,y) ~|~ \exists t \in [0,1) \text{ such that } \left( (x,t),
    (y,t) \right) \in R\} \] 
is a DT-bisimulation.
\end{thm}

\begin{proof}
First note that $R'$ is indeed an equivalence as the time-coherence of $R$
guarantees that $R'$ is transitive. 

Let us consider $x R' y$, i.e. we have that $(x,t) R (y,t)$ for some $t \in
[0,1)$ (note that it is true in fact for all $t \in [0,1)$ using
time-coherence). 

Let us prove that for any atomic proposition $A$, $\chi_A (x) = \chi_A
(y)$. Define $B_A = \{ \omega ~|~ \omega (0) \in A\}$. The set $B_A$ is
$obs$-closed, which means that $\mathbb{P}^{(x,t)} (B_A) =
\mathbb{P}^{(y,t)} (B_A)$. Furthermore, $\mathbb{P}^{(x,t)} (B_A) =
\chi_A(x)$ (and similarly for $y$) which proves our first point. 

Consider $B'$ an $R'$-closed set. Define $B = \{ (z, t+1) ~|~ z \in
B'\}$. This set $B$ is $R$-closed: consider $(z_1, t+1) R (z_2, t+1)$ and
$(z_1, t+1) \in B$. By definition of $R'$, $z_1 R' z_2$ and by definition
of $B$, $z_1 \in B'$. The set $B'$ is $R'$-closed, and therefore $z_2 \in
B'$ and hence $(z_2, t+1) \in B$. 

Since the set $B$ is $R$-closed, we have that $P_1((x,t), B) = P_1 ((y,t),
B)$. By definition of $P_t$, we get that $\tau(x,B) = \tau(y,B)$ which
concludes the proof that $R'$ is a DT-bisimulation. 
\end{proof}

These results can be summed up in the following theorem relating bisimulation and DT-bisimulation.

\begin{theorem}
Two states $x$ and $y$ (in the LMP) are DT-bisimilar if and only if for all $t \in [0,1)$, the states $(x,t)$ and $(y,t)$ (in the Feller-Dynkin process) are bisimilar.
\end{theorem}

\section{Conclusion}

We have given two definitions of bisimulation, one as an equivalence
relation and the other as a cospan of morphisms respecting the dynamics of
the process.  We have also studied many examples; the full version of this
conference submission contains many more examples.  It would be interesting
to know if (and under what conditions) an equivalence satisfying conditions
\ref{init2} and \ref{ind1} is a bisimulation.

However, there are many aspects to explore, as suggested by previous work
on step-based systems.  First of all is a quantitative description of
bisimulation through the definition of some metrics on the state
space. This was done in two ways.
\begin{itemize}
\item Either through the definition of a set of $[0,1]$-valued functions
  that can be viewed as experiments performed on the system
  (see~\cite{Desharnais04,vanBreugel05}). We are hoping that this set could
  be obtained by looking at hitting times and occupation times.
\item Or as a fixed-point of some operators on metrics
  (see~\cite{vanBreugel05}). Such a fixed-point metric was defined on jump
  processes in~\cite{Gupta04,Gupta06}. It seems that some of the details in
  that work are possibly incorrect, so we hope to fix those details and to
  adapt similar ideas to our framework.
\end{itemize}

Another important and interesting question is that of approximations
(see~\cite{Desharnais03,Danos03c,Chaput14}) of our Markov processes.  Here
we will undoubtedly face new subtleties as we will have to cope with both
spatial and temporal limits.

Finally, a fundamental result in this area is the logical characterization
of bisimulation~\cite{Benthem76,Hennessy80} which was also extended to the
probabilistic case~\cite{Desharnais02}.  We hope to be able to provide such
a logic for continuous-time processes based on the set of $[0,1]$-valued
functions used to obtain a bisimulation metric.  A game interpretation of
bisimulation could also be provided~\cite{Fijalkow17}.  Perhaps some
interesting insights could also come from nonstandard
analysis~\cite{Fajardo17} where there is also a notion of equivalence but
one which is quite different from bisimulation.  In that work the notion of
adapted spaces is fundamental.

\section*{Acknowledgements} This research has been supported by a grant
from NSERC.

\bibliographystyle{alpha}
\bibliography{main}

\newcommand{\etalchar}[1]{$^{#1}$}
\begin{thebibliography}{vBMOW05}

\bibitem[ACH{\etalchar{+}}95]{Alur95}
R.~Alur, C.~Courcoubetis, N.~Halbwachs, T.A. Henzinger, P.-H. Ho, X.~Nicollin,
  A.~Olivero, J.~Sifakis, and S.~Yovine.
\newblock The algorithmic analysis of hybrid systems.
\newblock {\em Theoretical Computer Science}, 138:3--34, 1995.

\bibitem[AD94]{Alur94}
R.~Alur and D.~Dill.
\newblock A theory of timed automata.
\newblock {\em Theoretical Computer Science}, 126:183--235, 1994.

\bibitem[BDEP97]{Blute97}
R.~Blute, J.~Desharnais, A.~Edalat, and P.~Panangaden.
\newblock Bisimulation for labelled {Markov} processes.
\newblock In {\em Proceedings of the Twelfth IEEE Symposium On Logic In
  Computer Science, Warsaw, Poland.}, 1997.

\bibitem[Bil08]{Billingsley08}
Patrick Billingsley.
\newblock {\em Probability and measure}.
\newblock John Wiley \& Sons, 2008.

\bibitem[Bob05]{Bobrowski05}
Adam Bobrowski.
\newblock {\em Functional analysis for probability and stochastic processes: an
  introduction}.
\newblock Cambridge University Press, 2005.

\bibitem[CDPP14]{Chaput14}
Philippe Chaput, Vincent Danos, Prakash Panangaden, and Gordon Plotkin.
\newblock Approximating {M}arkov processes by averaging.
\newblock {\em J. ACM}, 61(1):5:1--5:45, January 2014.

\bibitem[DDLP06]{Danos06}
Vincent Danos, Jos{\'e}e Desharnais, Fran{\c c}ois Laviolette, and Prakash
  Panangaden.
\newblock Bisimulation and cocongruence for probabilistic systems.
\newblock {\em Information and Computation}, 204(4):503--523, 2006.

\bibitem[DDP03]{Danos03c}
Vincent Danos, Jos{\'e}e Desharnais, and Prakash Panangaden.
\newblock Conditional expectation and the approximation of labelled {M}arkov
  processes.
\newblock In Roberto Amadio and Denis Lugiez, editors, {\em CONCUR 2003 -
  Concurrency Theory}, volume 2761 of {\em Lecture Notes In Computer Science},
  pages 477--491. Springer-Verlag, 2003.

\bibitem[DEP02]{Desharnais02}
J.~Desharnais, A.~Edalat, and P.~Panangaden.
\newblock Bisimulation for labeled {Markov} processes.
\newblock {\em Information and Computation}, 179(2):163--193, Dec 2002.

\bibitem[DGJP03]{Desharnais03}
J.~Desharnais, V.~Gupta, R.~Jagadeesan, and P.~Panangaden.
\newblock Approximating labeled {Markov} processes.
\newblock {\em Information and Computation}, 184(1):160--200, July 2003.

\bibitem[DGJP04]{Desharnais04}
Jos\'ee Desharnais, Vineet Gupta, Radhakrishnan Jagadeesan, and Prakash
  Panangaden.
\newblock A metric for labelled {Markov} processes.
\newblock {\em Theoretical Computer Science}, 318(3):323--354, June 2004.

\bibitem[DP03]{Desharnais03b}
Jos\'ee Desharnais and Prakash Panangaden.
\newblock Continuous stochastic logic characterizes bisimulation for
  continuous-time {M}arkov processes.
\newblock {\em Journal of Logic and Algebraic Progamming}, 56:99--115, 2003.
\newblock Special issue on Probabilistic Techniques for the Design and Analysis
  of Systems.

\bibitem[Dud89]{Dudley89}
R.~M. Dudley.
\newblock {\em Real Analysis and Probability}.
\newblock Wadsworth and Brookes/Cole, 1989.

\bibitem[Ein05]{Einstein1905}
A.~Einstein.
\newblock The theory of the brownian movement.
\newblock {\em Ann. der Physik}, 17:549, 1905.

\bibitem[FK17]{Fajardo17}
Sergio Fajardo and H~Jerome Keisler.
\newblock {\em Model theory of stochastic processes}, volume~14 of {\em Lecture
  Notes in Logic}.
\newblock Cambridge University Press, 2017.

\bibitem[FKP17]{Fijalkow17}
Nathanael Fijalkow, Bartek Klin, and Prakash Panangaden.
\newblock The expressiveness of probabilistic modal logic revisited.
\newblock In {\em Proceedings of the 44th International Colloquium on Automata
  Languages and Programming}, 2017.

\bibitem[GJP04]{Gupta04}
Vineet Gupta, Radhakrishnan Jagadeesan, and Prakash Panangaden.
\newblock Approximate reasoning for real-time probabilistic processes.
\newblock In {\em The Quantitative Evaluation of Systems, First International
  Conference QEST04}, pages 304--313. IEEE Press, 2004.

\bibitem[GJP06]{Gupta06}
Vineet Gupta, Radha Jagadeesan, and Prakash Panangaden.
\newblock Approximate reasoning for real-time probabilistic processes.
\newblock {\em Logical Methods in Computer Science}, 2(1):paper 4, 2006.

\bibitem[HM80]{Hennessy80}
Matthew Hennessy and Robin Milner.
\newblock On observing nondeterminism and concurrency.
\newblock In Jaco de~Bakker and Jan van Leeuwen, editors, {\em Automata,
  Languages and Programming}, volume~85 of {\em Lecture Notes in Computer
  Science}, pages 299--309. Springer Berlin / Heidelberg, 1980.

\bibitem[KS12]{Karatzas12}
Ioannis Karatzas and Steven Shreve.
\newblock {\em Brownian motion and stochastic calculus}, volume 113.
\newblock Springer Science and Business Media, 2012.

\bibitem[LS91]{Larsen91}
K.~G. Larsen and A.~Skou.
\newblock Bisimulation through probablistic testing.
\newblock {\em Information and Computation}, 94:1--28, 1991.

\bibitem[Mil80]{Milner80}
R.~Milner.
\newblock {\em A Calculus for Communicating Systems}, volume~92 of {\em Lecture
  Notes in Computer Science}.
\newblock Springer-Verlag, 1980.

\bibitem[Pan09]{Panangaden09}
Prakash Panangaden.
\newblock {\em Labelled {M}arkov Processes}.
\newblock Imperial College Press, 2009.

\bibitem[Par81]{Park81}
D.~Park.
\newblock Concurrency and automata on infinite sequences.
\newblock In {\em Proceedings of the 5th GI Conference on Theoretical Computer
  Science}, number 104 in Lecture Notes In Computer Science, pages 167--183.
  Springer-Verlag, 1981.

\bibitem[RW00]{Rogers00a}
L.~Chris~G. Rogers and David Williams.
\newblock {\em Diffusions, {M}arkov processes and martingales: {V}olume 1.
  {F}oundations}.
\newblock Cambridge university press, 2nd edition, 2000.

\bibitem[San09]{Sangiorgi09}
Davide Sangiorgi.
\newblock On the origins of bisimulation and coinduction.
\newblock {\em ACM Transactions on Programming Languages and Systems (TOPLAS)},
  31(4):15, 2009.

\bibitem[vB76]{Benthem76}
Johan van Benthem.
\newblock {\em Modal correspondence theory}.
\newblock PhD thesis, University of Amsterdam, 1976.

\bibitem[vBMOW05]{vanBreugel05}
Franck van Breugel, Michael Mislove, Joel Ouaknine, and James Worrell.
\newblock Domain theory, testing and simulation for labelled {M}arkov
  processes.
\newblock {\em Theoretical Computer Science}, 333(1-2):171--197, 2005.

\bibitem[Whi02]{Whitt02}
W.~Whitt.
\newblock {\em An Introduction to Stochastic-Process Limits and their
  Applications to Queues}.
\newblock Springer Series in Operations Research. Springer-Verlag, 2002.

\end{thebibliography}
\end{document}